\documentclass[12pt]{article}
\usepackage{hyperref}
\newtheorem{theorem}{Theorem}
\newtheorem{corollary}{Corollary}
\newtheorem{assumption}{Assumption}
\newtheorem{lemma}{Lemma}
\newtheorem{proposition}{Proposition}
\newtheorem{definition}{Definition}

\newtheorem{remark}{Remark}
\usepackage[toc,page]{appendix}
\newenvironment{proof}{\noindent{\bf Proof:}}{\hfill\fbox{}\vspace*{1mm}}
\usepackage[usenames]{color}
\usepackage{amsmath}
\usepackage{amssymb}
\usepackage{pdfpages}
\usepackage{graphicx}

\usepackage{subfigure}
\usepackage{float}
\usepackage{booktabs}
\usepackage[table]{xcolor}
\definecolor{lightgray}{gray}{0.9}

\RequirePackage[normalem]{ulem} %DIF PREAMBLE
\providecommand{\DIFdeltex}[1]{{\protect\color{red}\sout{#1}}}                      %DIF PREAMBLE
\usepackage{xcolor}
\usepackage{float}
\usepackage{tikz}
\usepackage{changebar}
\newif\ifdiff
\difftrue
\ifdiff
  \newcommand{\del}[1]{\DIFdeltex{#1}}

\else
  \newcommand{\del}[1]{}

\fi

\begin{document}
\title{\bf Optimal   Liquidation Problems     in a  Randomly-Terminated   Horizon }
\author{Qing-Qing Yang
\thanks{ Advanced Modeling and Applied Computing Laboratory,
Department of Mathematics, The University of Hong Kong, Pokfulam
Road, Hong Kong. E-mail: kerryyang920910@gmail.com.}
\and Wai-Ki Ching
\thanks{
Advanced Modeling and Applied Computing Laboratory,
Department of Mathematics,
The University of Hong Kong, Pokfulam Road, Hong Kong.
Hughes Hall, Wollaston Road, Cambridge, U.K.
School of Economics and Management,
Beijing University of Chemical Technology,
North Third Ring Road, Beijing, China.
E-mail: wching@hku.hk. }
\and Jia-Wen Gu
\thanks{Corresponding author. Department of Mathematics, Southern University of Science and Technology, Shenzhen, China. E-mail: jwgu.hku@gmail.com.}
\and Tak Kwong Wong
\thanks{ Department of Mathematics, The University of Hong Kong, Pokfulam
Road, Hong Kong. E-mail: takkwong@maths.hku.hk}
}
\date{\today}
\maketitle
\maketitle
\begin{abstract}
In this paper, we study  optimal liquidation problems in a randomly-terminated horizon.
We consider the liquidation of a large single-asset portfolio with the aim of minimizing a combination of volatility risk and transaction costs arising from permanent and temporary market impact.
Three different scenarios are analyzed  under Almgren-Chriss's market impact model  to explore the relation between optimal liquidation strategies and potential inventory risk arising from the uncertainty of the liquidation horizon.
For cases where no closed-form solutions can be obtained,
we verify comparison principles for viscosity solutions
and characterize the value function as the unique viscosity solution of the associated
Hamilton-Jacobi-Bellman (HJB) equation.
\end{abstract}



\section{Introduction}

Understanding trade execution strategies is a key issue for financial market practitioners and has attracted growing attention from the academic researchers. An important problem faced by equity traders is how to liquidate large  orders.  Different from small orders, an immediate execution of large orders is often impossible or at a very high cost due to insufficient liquidity.
A slow liquidation process, however,  is often costly, since it may involve undesirable  inventory risk.
Almgren and Chriss \cite{AC01} provided one of the early studies on the optimal execution strategy  of large trades, taking into account the volatility risk and liquidation  costs.
In order to produce tractable and analytical results, they set  the market impact cost per share to be linear in the rate of  trading.
Schied and Schoneborn \cite{SS09} considered the infinite-horizon optimal portfolio liquidation problem for a von Neumann-Morgenstern investor under the liquidity model of Almgren \cite{A03},
in which a power law cost function was introduced to determine optimal  trading strategies.
However, most of the literature on optimal liquidation strategies mainly considered
a known pre-determined time horizon or infinite horizon.
The case of unknown (or more precisely, randomly-terminated) time horizon is not fully addressed.
In some situation, it is more realistic to assume that the liquidation horizon  depends on some   stochastic factors of the model.
For example, some financial markets adopt the circuit-breaking mechanism,
which makes the horizon of the investor  subject to the  stock price movement.
Once the stock price hits the daily limits, all transactions of the stock will be suspended.

In this paper,  we consider a randomly-terminated time horizon under three different scenarios that an  agent might encounter in a financial market.
 Almgren-Chriss's market impact model is employed  to describe  the underlying asset price:
 $$
  \left\{
  \begin{array}{lll}
 \displaystyle  dS_t= f(\theta_t) dt+\sigma dW^S_t,\\
  \displaystyle \widetilde{S}_t=g(\theta_t)+S_t,\\
  \end{array}
  \right.
$$
where the constant   $\sigma>0$ is the absolute volatility of the asset price $S_t$,  $W^S_t$ is an one-dimensional standard Brownian motion,  $\widetilde{S}_t$ is the actual transaction price, $\{\theta_t, t\ge 0\}$ is an admissible control process,  $f(\theta_t)$ and $g(\theta_t)$ represent, respectively, the permanent and temporary components of the market impact.
We consider the liquidation problem of a large single-asset portfolio with the aim of minimizing a combination of volatility risk and transaction costs arising from permanent and temporary market impact.

We first  consider the case with a pre-determined time horizon $T$, which can be used as a  benchmark for  other cases with randomly-terminated time horizon.
In general, it is required that a liquidation strategy $\theta_t$
should satisfy  the {\it hands-clean} condition:
$$
X_T=X_0-\int_0^T\theta_tdt=0,
$$
where $X_t$ is  the number of shares held by the trader at the time $t$.
We  first work on a subclass of {\it deterministic} controls,
which do not allow for inter-temporal updating, satisfying the {\it hands-clean} condition.
Obviously, the deterministic strategy obtained in the subclass might be no longer optimal
when taking into account the entire class of admissible controls.
We then temporarily relax the {\it hands-clean} condition, and allow an immediate final liquidation (if necessary) so that the number of shares owned at the time
$t=T$ is $X_T=0$.
We employ the  dynamic programming (DP) approach to solve the stochastic control
problems and prove that the optimal liquidation strategy actually  converges to the deterministic strategy when the transaction cost involved by
liquidating the outstanding position $X_{T-}$  approaches to infinity.

We then move to analyze the randomly-terminated cases.
Two different scenarios are analyzed to shed light on the relationship
between  liquidation strategies and potential position risk arising from
the uncertainty of the time horizon.
First, we consider the scenario where the  liquidation process  is
terminated by an exogenous trigger event.
We model the occurrence time of a trigger event to be random and its hazard rate process is given by $\{l(t),t\ge0\}$.
Once this event occurs, all liquidation processes will be forced to suspend.
Compared with the case  without trigger event, agents facing the scenario  that  an exogenous trigger event might occur during the trading horizon would like to accelerate the rate of liquidating to reduce their exposure to potential position risk and eventually in a smaller position when the trigger event occurs.
Their strategy has  a steeper   gradient and is more ``convex'' when compared with those who are not threatened by this trigger event.
Second, we consider  the case when the liquidation process is subject to  counterparty risk.
Different from the exogenous trigger event setting, information set available to the counterparty risk modeler is more refined in terms of predictability.
To model counterparty risk, %two classes of models exist: structural models
%and reduced form models.
%Structural models  originated from Black and Scholes \cite{BS1},
%and Merton \cite{Merton}.  The reduced-form models originated from Jarrow and Turnbull  \cite{J92}, and  subsequently studied by Jarrow and Turnbull \cite{J95}, %Duffie and Singleton \cite{DS99} among others.
we adopt the structural firm value approach, originated from Black and Scholes \cite{BS1},
and Merton \cite{Merton}, and   let    the firm's asset  value follow a geometric Brownian motion:
$$
\displaystyle \frac{dY_t}{Y_t}=\beta dt+\xi dW^Y_t.
$$
The incorporation of counterparty risk into the study of optimal liquidation does not come without cost.
In order to examine its impact on optimal trading strategies,
we have to introduce and employ viscosity solutions.
By verifying the comparison principles for viscosity solutions,
we characterize the value function as the unique viscosity solution of the associated
Hamilton-Jacobi-Bellman (HJB) equation.
This equation can be numerically solved.
We further analyze the effectiveness of the numerical method  and
illustrate that the computational error is sufficiently small.
%Third, we discuss the liquidation problem in a market with a circuit-breaking mechanism.
%In such a market there are rules  to regulate trade practices.
%For example, the daily transactions are subject to limits on price movements.
%When the stock price falls (rises) by their daily limit, all market transactions on this stock will be suspended.
%Similar to the second model, we characterize the value function as the unique viscosity solution  and this can be used to obtain further results.

The remainder of this paper is structured as follows.
The background and basic models of an agent's liquidation problem are introduced  in Section 2.
Section 3 discusses typical liquidating problems  under the  benchmark model. 
%including deterministic control  and stochastic  control.
In Sections 4--5, we discuss two different scenarios with randomly-terminated time horizons.
Viscosity solution approach is adopted in these sections to study in great generality stochastic control problems. By combining these results with comparison principles, we characterize the value function as the unique viscosity solution of the associated
dynamic programming  equation, and this can then be used to obtain further results.
Finally, concluding remarks are given in Section 6. For the sake of self-containedness, we provide the technical proofs in the Appendix.

\section{Problem Setup }

In this section, we first describe the market environment of the agent.
We then present a market impact model 
%including both temporary and permanent impacts   
to discuss the optimal  liquidating  problem.

\subsection{The Market Environment and Market  Impact Model}

The agent starts at $t=0$ and  has to  liquidate a large position in a risky asset
by  time $T$.
This terminal time can be either deterministic or random,
depending on the scenario that   the agent is facing.
For simplicity,  we assume that the agent withholds  the liquidation proceeds.
In other words,  he/she does  not  deposit the liquidation proceeds in his/her
money market account.
At any time $t\in[0,T]$, we adopt the following notations for the agent's portfolio:
\begin{description}
\item[(i)] $V_t=C_t+X_tS_t, \mbox{portfolio value}$;
\item[(ii)] $C_t, \mbox{balance of risk free bank account}$;
\item [(iii)] $X_t, \mbox{number of shares of underlying asset}$;
\item[(iv)] $S_t, \mbox{price of the underlying risky asset}$.
\end{description}
The initial conditions are   $ C_0=0$, $S_0=s$,  and $X_0=Q$.\\

Suppose the risky asset can be continuously liquidated during the trading horizon,
namely,  there is always sufficient liquidity for their execution\footnote{For simplicity, the  transaction fees   will not be considered in this paper.}.
Let  $\{\theta_t\}_{t\in[0,T]}$ denote the liquidation process.
The shares held by the trader at any time $t\in[0,T]$ can be written as follows:
$$
X_t=Q-\int_{0}^t\theta_{u}du.
$$
We
consider a probability space $(\Omega, \mathcal{F}, \mathcal P)$  endowed with a filtration $\{\mathcal{F}_t\}_{t\ge0}$.
%basing on which we introduce the definition of an admissible control process.

\begin{definition}
A stochastic process $\theta(\cdot)=\{\theta_u, 0\le u\le T\}$   is called an admissible control process if  all of the following conditions hold:
\begin{description}
\item[(i)] {\bf(Adaptivity)} For each $t\in[0,T]$, $\theta_t$ is $\mathcal F_t$-adapted;

\item[(ii)]{\bf(Non-negativity)} $\theta_t\in \mathbb R_+$, where $\mathbb R_+$ is the set of nonnegative real values;

\item[(iii)]{\bf(Consistency)}
$$
\displaystyle  \int_0^T \theta_tdt\le Q;
$$
\item[(iv)] {\bf(Square-integrability)}
$$
\displaystyle \mathbb E\left[\int_0^T|\theta_t|^2dt\right]<\infty;
$$
\item[(v)] {\bf($L_\infty$-integrability)}
$$
\mathbb E\left[\max_{0\le t\le T}|\theta_t|\right]<\infty.
$$
\end{description}
Furthermore, denote  $\Theta_t$ as the collection of admissible controls with respect to the initial time $t \in[0,T)$ and $\widehat\Theta_t$ as the collection of controls only satisfying condition (i), (iv) and (v).
\end{definition}

We assume that the risky  asset exhibits a price impact due to the feedback effects
of the agent's liquidation strategy.
For any given admissible control $\theta(\cdot)\in\Theta_0$,
the market mid-price  of the stock  is assumed to follow the   dynamics:
\begin{equation}\label{underlying0}
  dS_t=  f(\theta_t) dt+\sigma dW^S_t,\\
\end{equation}
where $\{W^S_t\}$ is a standard Brownian motion with filtration $\{\mathcal F_t\}$,  the constant $\sigma>0$ is the absolute volatility of the asset  price, and  $f(\cdot)$ is  the permanent component of the market impact.
%In Eq. (\ref{underlying0}), the drift term is set to be zero, which means that we expected no obvious trend in the stock price. The total time is usually one day or less, the drift is
%generally not significant over such a short trading horizon.
For simplicity, we  further assume that $f$ is   time homogeneous, namely,   $f(\cdot)$ is independent of $t$.

Generally speaking, the actual transaction price $\widetilde{S}_t$ is
not always the same as  the market  mid-price $S_t$, since the market is not perfectly liquid, see, for example, Almgren and Chriss  \cite{AC01}.
We assume $\widetilde{S}_t= S_t + g(\theta_t)$ and call $g(\theta)$ the temporary price impact.
Intuitively, the function $g(\cdot)$ captures quantitatively how the limit order books available in the market are eaten up at different levels of trading speeds. \\

\noindent
{\bf Assumption 0.} The price dynamics follow a simple Almgren-Chriss linear market impact  model
(see, Almgren and Chriss  \cite{AC01}):
$$
\begin{array}{l}
\displaystyle f(\theta)=-\eta\cdot  \theta\quad \mbox{and}\quad
\displaystyle g(\theta)=-\nu\cdot \theta,\\
\end{array}
$$
where $\eta$ and $\nu$ are  positive constants.

An agent who holds the stock receives the capital gain or loss due to stock price movements.
Thus, if   the agent's position is marked to market using the book value,
ignoring market impact that would be incurred in converting these shares into cash,
at any time $t$, the agent's portfolio value  $V_t=C_t+X_tS_t$ satisfies
\begin{equation}\label{dd}
\left\{
\begin{array}{l}
 V_0=Q\cdot s\\
\displaystyle  dV_t=
\displaystyle  (\widetilde{S}_t-S_t)\theta_tdt+ X_tdS_t.
\end{array}
\right.
\end{equation}
At any time $t\in[0,T)$ before the end of trading,
$$
 \begin{array}{lll}
V_t&=&V_0+\displaystyle \int_0^t  (\widetilde{S}_u-S_u)\theta_u  du+\int_0^tX_udS_u\\
 &=&\displaystyle V_0+ \int_0^t \left[(\widetilde{S}_u-S_u)\theta_u  +X_u f(\theta_u)\right]du +\int_0^t\sigma X_udW^S_u.
 \end{array}
 $$

 \subsection{Hands-clean Condition}

Let us recall that  our task is to liquidate a large-size position  by the  time $T$.
Generally speaking, it is required that the {\it hands-clean} condition should be satisfied:
\begin{equation}\label{hand}
X_T=X_0-\int_0^T\theta_tdt=0.
\end{equation}
This technical condition, however,  introduces some unexpected properties to the  stochastic control problem. To tackle this problem,  we    temporary relax the {\it hands-clean} condition and allow an inmmediate final liquidation (if necessary) so that the number of shares owned at $t=T$ equals zero. That is,
given the state variables $(S_t, C_t, X_t)$ at  the instant before the end of trading $t=T-$,  if $X_{T-}\ne 0$, then we will  have an immediate final liquidation so that $X_T=0$. The liquidation proceeds $C_T$ after this final trade is
$$
C_T=C_{T-}+X_{T-}\left(S_{T-}-\mathcal C^o(X_{T-})\right),
$$
where $\mathcal C^o(X_{T-})=\phi X_{T-}$, for some  constant $\phi>0$,
is  the  cost involved from liquidating the outstanding position $X_{T-}$.
%(if the position is not completely liquidated at time $T-$, then the trader has to execute  the non-executed orders at the market price with certain clearing fee, $\phi$ per share).
Thus, we have
$$
V_T=V_{T-}-\phi X^2_{T-}.
$$
The gain/loss from liquidating  the outstanding    position, $R_T=V_T-V_0$, is given by
\begin{equation}\label{1p}
R_T= \int_{[0, T)} \left[(\widetilde{S}_t-S_t)\theta_t +X_tf(\theta_t)\right]dt +\int_{[0, T)}\sigma X_tdW^S_t-\phi X^2_{T-}.
\end{equation}

\subsection{Performance Criterion}

Under the normal circumstance, investors are risk averse and demand a higher return for a riskier investment.
The mean-variance criterion  is popular for    taking both return and risk into account.
However, the mean-variance criterion  may induce a potential problem of time-inconsistency, namely, planned and implemented policies are different.
As mentioned in Rudloff et al. \cite{Rudloff},
a major reason for developing dynamic models instead of static ones is the fact that
one can incorporate the flexibility of dynamic decisions to improve the objective function.
Time-inconsistent criteria are generally not favorable to introduce in the study,
since the associated policies are sub-optimal.

To take both return and risk into account, instead of adopting
the mean-variance criterion,  we are most interested in the mean-quadratic
optimal agency execution strategies,
as they are proved to be time-consistent in \cite{AC01, FKT12,PA1}.
In this section, we will introduce the quadratic variation and the corresponding objective function as follows.

\subsubsection{Quadratic Variation}

Formally, the {\it quadratic variation} of  the portfolio value $V$ on $[0,T)$ is defined to be
\begin{equation}\label{qd}
[V, V]([0,T))=\int_{[0, T)}\sigma^2X^2_tdt.
\end{equation}
From the interpretation of Eq. (\ref{qd}), minimizing quadratic variation corresponds to minimizing volatility in the portfolio value process.

\subsubsection{Objective Function }
Let  $\gamma>0$ be a constant corresponding to the risk aversion.
Then the  agent's objective is  to find the optimal control for
\begin{equation}\label{ob}
\begin{array}{lll}
&&\displaystyle \max_{\theta(\cdot)\in \Theta_0}\mathbb E[R_T-\gamma [V, V]([0,T))]\\
&=&\displaystyle \max_{\theta(\cdot)\in \Theta_0}\mathbb E\bigg[\int_{[0, T)} \left[(\widetilde{S}_t-S_t)\theta_t  +X_t f(\theta_t)-\gamma\sigma^2X_t^2\right]dt -\phi X^2_{T-}\bigg].
\end{array}
\end{equation}

\section{The Benchmark Model for Optimal Liquidation (Model 1)}

\begin{assumption}\label{assumption1}
The liquidation horizon $T$ is a finite-valued, pre-determined, and positive  constant.
\end{assumption}

 In this section, we present our benchmark model   under  Assumption 0 for the optimal liquidation problem.
We first work on a subclass of {\it deterministic} controls\footnote{Controls that do not allow for inter-temporal updating.} satisfying the {\it hands-clean} condition (\ref{hand}), and then move to  the
dynamic programming (DP) approach considering over the entire class of admissible controls.
We prove that when the  transaction cost involved by liquidating the outstanding position $X_{T-}$ approaches to infinity, the optimal liquidation strategy obtained from DP approach  converges to the deterministic one.

\subsection{Deterministic Control}

Let us first consider the case in which $\theta(\cdot)$ ranges only over the sub-class $\Theta^{det}_0$ of {\it deterministic} strategies in $\Theta_0$ satisfying the {\it hands-clean} condition
$$
\int_0^T\theta_tdt=Q.
$$
That is,   $X_{T-}=0$, and   the agent's objective  is to    find the optimal strategy for
\begin{equation}\label{determine}
\begin{array}{lll}
&&\displaystyle \max_{\theta(\cdot)\in \Theta^{det}_0}\mathbb E \int_{[0, T)} \left[(\widetilde{S}_t-S_t)\theta_t +X_t f(\theta_t)-\gamma\sigma^2X_t^2\right]dt\\
&=&\displaystyle \max_{\theta(\cdot)\in \Theta^{det}_0}\mathbb E \int_{[0, T)} \left[g(\theta_t)\theta_t +X_t f(\theta_t)-\gamma\sigma^2X_t^2\right]dt.
\end{array}
\end{equation}
The  cost function  of  the deterministic control problem  (\ref{determine}) is
$$
\mathcal H(X_t,\theta_t,\Lambda_t,t)\equiv  g(\theta_t)\theta_t+f(\theta_t)X_t-\gamma \sigma^2 X_t^2 -\Lambda_t\theta_t,
$$
where $\Lambda_t$ is the {\it Lagrange multiplier} (also called the
{\it adjoint state}). 
%The Lagrange multiplier provides the objective extension for including the state dynamics. 
The differential equation for the deterministic system is:
$$
\frac{dX_t}{dt}=-\theta_t \quad {\rm with} \quad X_0=Q.
$$
We assume the Hamiltonian $\mathcal H$ has continuous first-order derivatives in
state, adjoint state, and control variables, namely, $\{ X_t, \Lambda_t, \theta_t\}$.
Then the necessary conditions (also called {\it Hamilton's equation}) for having an interior optimum of the Hamiltonian
$\mathcal H$ at  $\{X_t^{det, *}, \Lambda_t^{det,*}, \theta_t^{det,*}\}$,  are given by
\begin{equation}\label{kl00}
\left\{
\begin{array}{rll}
\displaystyle \frac{d X_t^{det,*}}{dt}&=&\displaystyle \frac{\partial \mathcal H}{\partial \Lambda}\bigg|_{(X_t^{det,*},\theta_t^{det,*},\Lambda_t^{det,*},t)}\\
\displaystyle -\frac{d\Lambda_t^{det,*}}{dt}&=&\displaystyle \frac{\partial \mathcal H}{\partial X}\bigg|_{(X_t^{det,*},\theta_t^{det,*},\Lambda_t^{det,*},t)}\\
\displaystyle 0&=&\displaystyle  \frac{\partial \mathcal H}{\partial \theta}\bigg|_{(X_t^{det,*},\theta_t^{det,*},\Lambda_t^{det,*},t)}.
\end{array}
\right.
\end{equation}
It follows from the critical conditions in Eq. (\ref{kl00}) and Assumption 0 that
\begin{equation}\label{er0}
\left\{
\begin{array}{l}
\displaystyle   \ddot{X}_t^{det,*}=   \frac{\gamma \sigma^2 }{\nu}X_t^{det,*},\\
X_0^{det,*}=Q,\\
X_T^{det,*}=0.
\end{array}
\right.
\end{equation}
An explicit solution, which is unique according to Lasota and Opial \cite{LO67},
is given by
\begin{equation}\label{op12}
\left\{
\begin{array}{l}
 \displaystyle X_t^{det,*}=Q\frac{\sinh\big(\sqrt{\frac{\gamma\sigma^2}{\nu}}(T-t)\big)}{\sinh\big(\sqrt{\frac{\gamma\sigma^2}{\nu}}T\big)},\\

\displaystyle   \theta_t^{det,*}= Q\sqrt{\frac{\gamma\sigma^2}{\nu}}\frac{\cosh\big(\sqrt{\frac{\gamma\sigma^2}{\nu}}(T-t)\big)}{\sinh\big(\sqrt{\frac{\gamma\sigma^2}{\nu}}T\big)}.
  \end{array}
	\right.
\end{equation}
\quad\\

There is  a very interesting phenomenon in the {\it deterministic} control problem: the solution (\ref{op12}) has nothing to do with the permanent price impact $\eta$.
%In other words, no matter what the value of $\eta$ is, one cannot see the effect of $f(\theta)=-\eta\theta$.
%The rationale behind this is that the  cumulative impact cost,
%under the linear price impact model,  does not depend on the time taken or strategy used
%to execute the liquidation.
If a position of size $Q$ units with initial market price $s$ is fully liquidated by time
$T$, i.e. $X_T=0$, the expected value of the resulting cash becomes
$$
\begin{array}{lll}
\displaystyle\mathbb E\left[ \int_0^T\widetilde S_t\theta_tdt\right]
&=&\displaystyle \mathbb E\left[\int_0^TS_t\theta_tdt-\nu\int_0^T\theta_t^2dt\right]\\
&=&\displaystyle Q\cdot s+\mathbb E\left[\int_0^TX_tdS_t-\nu\int_0^T\theta_t^2dt\right]\\
&=&\displaystyle Q\cdot s-  \underbrace{\mathbb E\bigg[\nu\int_0^T\theta_t^2dt\bigg]}_{ \mbox{(temporary impact cost)}}-\underbrace{\frac{1}{2}\eta Q^2.}_{\mbox{(permanent  impact cost)}}\\
\end{array}
$$
Clearly,  the permanent  impact cost
is independent of the time taken or strategy used to execute the liquidation.
%That is, we cannot avoid a permanent  impact cost of $\frac{1}{2}\eta Q^2$,
%quadratic in portfolio size.

\subsection{Dynamic Programming Approach}

Obviously, if  we are allowed  to  update dynamically, namely,
replacing $\Theta^{det}_0$ by the entire class of admissible strategies $\Theta_0$,
then one will be able to further  improve his/her performance.
In this section, we consider a stochastic  approach.
We  employ the  DP method to   solve  the stochastic control problem (\ref{ob}).
This approach yields a   Hamilton-Jacobi-Bellman (HJB)  equation.
When this HJB equation  can be solved by an explicit   smooth solution,
the verification theorem then validates the optimality of the candidate solution to the HJB equation.
For more details about the verification theorem, we refer interested readers to
Pham  \cite{Pham09} (Chapter 3), \O ksendal  \cite{O1} (Chapter 11),
and \O ksendal and Sulem  \cite{O2} (Chapter 3).\\

%\subsubsection{Hamilton-Jacobi-Bellman (HJB)  Equation}

Let $U(t, q)$  be the optimal value function beginning at  a time $t\in[0, T)$
with  initial value  $X_t=q$, namely\footnote{It is worth noting that the value function $U$ does not depend explicitly on the stock price $S_t$.},
\begin{equation}\label{e1}
U(t,q)=\max_{\theta(\cdot)\in \Theta_t}\mathbb E\left[ \int_{[t, T)} \big[-\nu\theta^2_u- \eta X_u\theta_u -\gamma \sigma^2X_u^2\big]du-\phi X_{T-}^2\Big|\mathcal F_t\right].
 \end{equation}
Temporarily assuming  that $U(t, q)\in \mathcal C^{1,2}([0,T)\times(0,+\infty) ).$\footnote{$\mathcal C^{1,2}([0,T)\times  (0,+\infty))$ is the space of functions $f(t, q)$ which is continuously differentiable in $t$, and twice continuously differentiable in   $q$.}
 From the DP principle, $U$ must satisfy   the following  HJB equation:
\begin{equation}\label{HJB}
\left\{
\begin{array}{lll}
 \displaystyle \partial_t U   -\gamma  \sigma^2  q^2  -\min_{\theta_t\in\Theta_t}\Big\{ \nu\theta^2_t+(\eta q  +\partial_{q}U)\cdot \theta_t\Big\}=0\\
\displaystyle U(T-, q)=-\phi q^2.
\end{array}
\right.
\end{equation}

We remark that the optimization problem included in Eq.
(\ref{HJB}) is a constrained optimization problem with constraints:
(a1) $\theta_t\in \mathbb R_+$; and  (a2) $\int_{[0, T)} \theta_tdt\le Q$.
Generally speaking,  there is no straightforward method to solve  this kind of problems.
One  simple way to handle this problem  is to
consider the corresponding unconstrained optimization problem:
 \begin{equation}\label{HJB-1}
\left\{
\begin{array}{lll}
 \displaystyle \partial_t U   -\gamma  \sigma^2  q^2  -\min_{\theta_t\in\widehat\Theta_t}\Big\{ \nu\theta^2_t+(\eta q  +\partial_{q}U)\cdot \theta_t\Big\}=0\\
\displaystyle U(T-, q)=-\phi q^2,
\end{array}
\right.
\end{equation}
 and then verify that the obtained result indeed satisfies all the constraints. 
 %In this paper, we will use this method to identify the optimal trading strategy.
From the HJB equation, Eq. (\ref{HJB-1}), the optimal trading strategy without constraints is given by
$$
\theta_{t}^{\phi,*}=-\frac{1}{2\nu}\big(\partial_{q}U+\eta q \big).
$$
Thus the value function $U$ solves the following Ordinary  Differential Equation (ODE):
\begin{equation}\label{0}
\left\{
\begin{array}{lll}
  \displaystyle \partial_t U    -\gamma  \sigma^2  q^2  +\frac{1}{4\nu}  \big(\partial_{q}U+\eta q \big)^2=0\\
  \displaystyle U(T-,q)=-\phi q^2.
\end{array}
\right.
\end{equation}

\begin{theorem}\label{theorem1}
 There is at most one  $\mathcal C^{1,2 }([0,T)\times(0, \infty))$  solution to  Eq. (\ref{0}).
\end{theorem}

\begin{proof}
Let  $f_1$ and $f_2$ be two  $\mathcal C^{1,2 }([0,T)\times(0, \infty))$  solutions to  Eq. (\ref{0}). Define  $\widetilde{f} =f_1-f_2$.  Then the new function $\widetilde {f}$  satisfies the following Partial Differential Equation
(PDE):
$$
\left\{
\begin{array}{lll}
 \displaystyle \partial_t \widetilde{f}
+\frac{1}{4\nu}  \big[ \partial_{q}(f_1+f_2) +2\eta q \big]\partial_q \widetilde{f}=0\\
 \displaystyle \widetilde{f} (T-, q)=0.
\end{array}
\right.
$$
Since the evolution equation for $\widetilde {f}$ is linear and first-order, one can solve the above problem explicitly  by the method of characteristics, and find that $\widetilde{f}\equiv 0$ is the unique solution to this problem. As a result, $f_1\equiv f_2$.
\end{proof}
 \quad\\

To solve Eq. (\ref{0}), we consider an ansatz that is quadratic in the variable $q$:
$$
U(t,q)=a(t)+ b(t)q +c(t)q^2.
$$
%\begin{figure}[H]
%\begin{center}
%\includegraphics[width=4.6in, height=3.2in]{x1.pdf}
%\caption{ Function $(1-x)/(1+x)$.}
%\label{x}
%\end{center}
%\end{figure}
\noindent
According to Theorem 1, if  the above ansatz  is a solution of Eq. (\ref{0}),
then it must be  the unique solution.
Under this setting, the optimal liquidating strategy takes the following form:
$$
\theta_t^{\phi,*} = -\frac{1}{2\nu}\left\{ b(t)+[2c(t)+\eta]q\right\}.
$$
A direct substitution yields that the coefficients  $a(t)$,  $b(t)$ and $c(t)$
must satisfy the following ODEs:
\begin{equation}\label{ep}
\left\{
\begin{array}{l}
 \displaystyle \dot{c}(t)=\gamma\sigma^2-\frac{1}{4\nu}[2c+\eta]^2\\
 \displaystyle \dot{b}(t)= -\frac{1}{2\nu}b(t)[2c+\eta]\\
 \displaystyle \dot{a}(t)=-\frac{1}{4\nu}b^2
 \end{array}
 \right.
 \end{equation}
with  terminal conditions: $a(T)=0$, $b(T)=0$ and  $c(T)=-\phi$.
 Since System (\ref{ep}) is partially decoupled, we can  find the exact solution via direct integrations. As a result,
they are given by
\begin{equation}\label{r1}
\left\{
\begin{array}{l}
 \displaystyle c(t)=\frac{1}{2\xi}\left[\frac{\zeta e^{-4\gamma\xi\sigma^2(T-t)}-1}{\zeta e^{-4\gamma\xi\sigma^2(T-t)}+1}\right]-\frac{\eta}{2}\\
 \displaystyle
\displaystyle   b(t)=0\\
 \displaystyle
 a(t)=0
 \end{array}
 \right.
\end{equation}
where  the constants $\zeta$ and $\xi$ are  given  by
$$
\zeta=\frac{1-\xi(2\phi-\eta)}{1+\xi(2\phi-\eta)}
\quad {\rm and} \quad \xi=\frac{1}{2\sigma\sqrt{\gamma\nu}}.
$$
 \quad\\

It is worth noting that
\begin{equation}\label{opl}
\dot{X}^{\phi,*}_t=-\theta^{\phi,*}_t=\frac{1}{2\nu}[2c(t)+\eta]X^{\phi, *}_t,
\end{equation}
and that   $X^{\phi, *}_0=Q$.
Therefore,
$$
X^{\phi,*}_t=Q \cdot \exp\left(\frac{1}{2\nu}\int_0^t[2c(u)+\eta]du\right).
$$
As to the results obtained in this section, we have the following proposition.
\begin{proposition}\label{proposition1}
It is assumed that model parameters satisfy the condition:
\begin{equation}\label{cond1}
 2\phi>\eta+2\sigma\sqrt{\gamma\nu}.
\end{equation}
That is, market liquidity risk  dominates the potential arbitrage opportunity introduced by permanent impact and potential position risk involved by price fluctuations.
Then, $c(t)$ is a strictly  decreasing function in $t$ and $c(t)+\eta\le 0$ for any $t\in[0,T)$. Furthermore, we have that
\begin{description}
\item[(b1)] $\theta^{\phi, *}_t\ge0$, for any time $t\in[0,T)$; and that
\item[(b2)] $\int_{[0, T)}\theta^{\phi, *}_tdt\le Q$.
\end{description}
The obtained optimal trading strategy (\ref{opl}) is also the optimal trading strategy for the constrained problem.
\end{proposition}
\begin{proof}
Notice  that the graph of the function $c(t)$ depends on the coefficient $\zeta=(1-x)/(1+x)$, with $x=(2\phi-\eta)/2\sigma\sqrt{\gamma\nu}$.
Under Assumption (\ref{cond1}),    $x>1$, and hence $-1<\zeta<0$.
 Therefore,
  $$
  \frac{\partial c(t)}{\partial t}<0,
 $$
 i.e., $c(t)$ is a strictly  decreasing function in $t$, and
 $c(t)+\eta/2\le 0$ always holds for any $t\in[0,T)$.  Thus, we conclude that
$$
 \theta^{\phi, *}_t=-Q\frac{1}{2\nu}[2c(t)+\eta]e^{\frac{1}{2\nu}\int_0^t[2c(u)+\eta]du}\ge0,
$$
for any time $t\in[0,T)$, and that
$$
\int_{[0, T)}\theta_t^{\phi, *}dt=Q\Big[1-e^{\frac{1}{2\nu}\int_0^t[2c(u)+\eta]du}\Big]\le Q.
 $$
 \end{proof}
\quad\\

  Let $U_T(t,q)$ denote the  value function of the optimization problem (\ref{e1})  with time horizon $T$, then  for any $T_1>T_2>t$,  we have
 \begin{equation}\label{red}
 U_{T_1}(t,q)>U_{T_2}(t,q),
 \end{equation}
provided that the condition (\ref{cond1}) holds. This  is consistent  with the fact that  an investor's ability  to bear risk  relates to his/her  time horizon for investment\footnote{The ability to bear risk is measured mainly in terms of objective factors, such as time horizon,  expected income, and the level of wealth relative to liability.}.
%For example, an investor with a 20-year time horizon can be considered to have a
%greater ability to bear risk  than an investor with a 2-year horizon provided that all other factors remain the same.
%This difference is because over 20 years there is more scope for losses to be recovered or other adjustments to circumstances to be made than there is over 2 years.
\quad\\

\subsection{Relation between Deterministic  and Stochastic  Control}

\begin{theorem}\label{theorem2}
When the transaction fees  involved by liquidating the outstanding position $X_{T-}$  approaches to infinity,  the limit of the  optimal stochastic control process $(\theta^{\phi,*}_t)_{t\in[0,T)}$  satisfies the {\it hands-clean} condition and it converges (point-wise) to the optimal deterministic control process $(\theta^{det,*}_t)_{t\in[0,T)}$. Meanwhile, the optimal trajectory $X_t^{\phi,*}$     converges (point-wise) to the one determined   in the deterministic system $X_t^{det, *}$. That is,   as $\phi\to \infty$,  we have
\begin{enumerate}
\item $\displaystyle X^{\phi,*}_{T-}\to 0$;
\item $\displaystyle \lim_{\phi\to \infty} \theta^{\phi,*}_t= \theta^{det,*}_t$\quad point-wise;
\item $\displaystyle \lim_{\phi\to \infty}X_t^{\phi,*}= X_t^{det,*}$ \quad point-wise.
\end{enumerate}
\end{theorem}

\begin{proof} We complete the proof by the following two steps:
\begin{description}
\item[Step 1 (Hands-clean condition)]
We first  prove that,  as   $\phi \to \infty$, $X_{T-}^{\phi,*}\to 0$.
We note that
 $$
  X_{t}^{\phi, *}= Q\cdot \exp\left(\int_0^{t}\frac{1}{2\nu}[2c(u)+\eta]du\right).
 $$
 A simple calculation yields
\begin{equation}\label{proof1}
\begin{array}{l}
\displaystyle e^{ \int_u^t\frac{1}{2\nu}[2c(r)+\eta]dr}=\frac{\zeta e^{-4\gamma\xi\sigma^2(T-t)}+1}{\zeta e^{-4\gamma\xi\sigma^2(T-u)}+1}e^{-2\gamma\xi\sigma^2(t-u)}.
\end{array}
\end{equation}
As $\phi\to\infty$,  $\zeta\to -1$, and hence
$$
\begin{array}{rll}
X_{T-}^{\phi,*}&=& \displaystyle \frac{Q(\zeta+1)}{\zeta e^{-2\gamma\xi\sigma^2T}+e^{2\gamma\xi\sigma^2T}}\rightarrow 0.
  \end{array}
$$
\item[Step 2 (Convergence)]
We then prove that  as $\phi\to\infty$,
 \begin{itemize}
\item $\displaystyle \lim_{\phi\to \infty} \theta^{\phi,*}_t= \theta^{det,*}_t$\quad point-wise; and
\item $\displaystyle \lim_{\phi\to \infty}X_t^{\phi,*}= X_t^{det,*}$\quad point-wise.
\end{itemize}
 First, we have
 $$
 \lim_{\phi\to\infty} X_t^{\phi, *}= \lim_{\phi\to \infty}Qe^{ \int_0^t\frac{1}{2\nu}[2c(u)+\eta]du}=Q\frac{e^{2\gamma\xi\sigma^2(T-t)}-e^{-2\gamma\xi\sigma^2(T-t)}}{e^{2\gamma\xi\sigma^2T}-e^{-2\gamma\xi\sigma^2 T}}=X_t^{det,*}.
 $$
For any time $t\in[0,T)$,
 $$
 \lim_{\phi\to \infty} [2c(t)+\eta]=\frac{1}{\xi}\frac{e^{-4\gamma\xi\sigma^2(T-t)}+1}{e^{-4\gamma\xi\sigma^2(T-t)}-1}.
  $$
Thus,  we have
 $$
 \lim_{\phi\to\infty}\theta_t^{\phi,*} = \lim_{\phi\to\infty}-\frac{1}{2\nu}[2c(t)+\eta] X_t^{\phi,*}=
\displaystyle  \frac{Q}{2\nu\xi} \frac{e^{-2\gamma\xi\sigma^2(T-t)}+e^{2\gamma\xi\sigma^2(T-t)} }{e^{2\gamma\xi\sigma^2 T}-e^{-2\gamma\xi\sigma^2 T}  }=\theta^{det,*}_t.
$$
\end{description}
\end{proof}

 \quad\\

  In Figure \ref{ds}, we   illustrate how the transaction fees involved  by liquidating the outstanding position $X^{\phi, *}_{T-}$, $\phi |X^{\phi,*}_{T-}|^2$,   affects the agent's liquidating speed.
We chose the following values of the model parameters:
$T=1\; day$, $Q=100\;units$,  $\gamma=0.1$,  $\sigma=0.2$, $\eta=0.001$ and $\nu=0.003$.

\begin{figure}[H]
\begin{center}
\includegraphics[width=6in, height=3.5in]{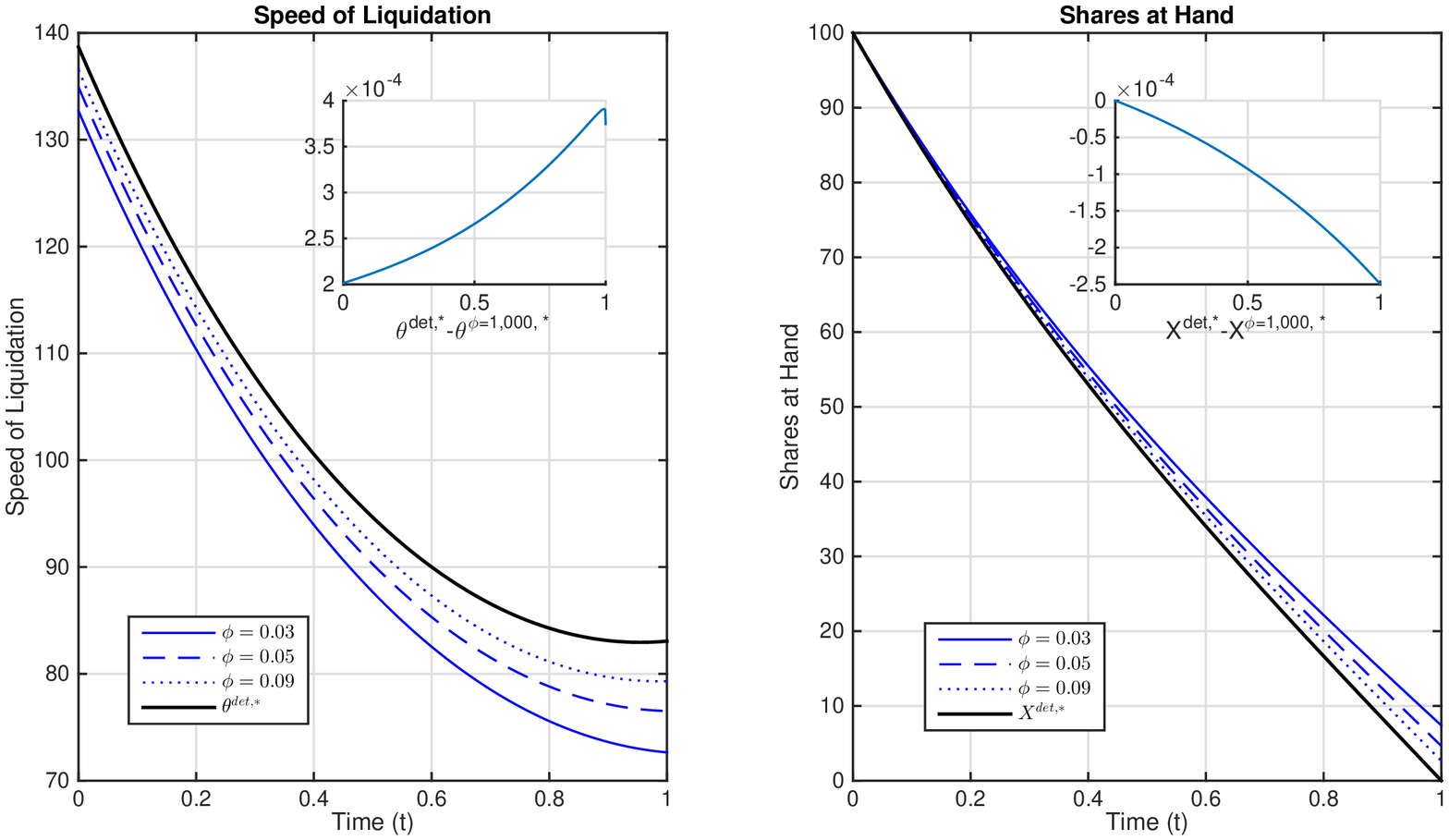}
\caption{Optimal deterministic control vs. Optimal stochastic controls
(model parameters: $T=1\; day$, $Q=100\;units$,  $\gamma=0.1,  \sigma=0.2,\eta=0.001,\nu=0.003$). }
\label{ds}
\end{center}
\end{figure}

Figure \ref{ds} illustrates that the speed of liquidation  which is  free of  
{\it hands-clean} condition is always slower than that under  the constraint of {\it hands-clean} condition.
As  the transaction fees involved  by liquidating the outstanding position $X^{\phi,*}_{T-}$
increases (namely, as $\phi$ increases),   the agent's liquidating speed increases, indicating that the optimal stochastic control moves closer to the optimal deterministic control.
%This effect is quite significant.
%We see that even  for a very small value of $\phi$, the remaining outstanding position the instant before the end of trading $X^{\phi,*}_{T-}$ can be small enough to be neglected.
The embedded subfigures in Figure 1 show, respectively,  the differences between the deterministic and stochastic liquidating strategies and the corresponding trajectories with
$\phi=1,000$.
Both of them  are of magnitude $10^{-4}$.

\section{Optimal Liquidation Strategy Subject to an Exogenous Trigger Event (Model 2)  }

In this section, we extend our results to models with an  exogenous event, which does not depend on
the information structure $\{\mathcal F_t\}_{t\ge0}$.

\begin{assumption}\label{assumption2}
The liquidation process will be suspended, if an exogenous trigger event occurs.
\end{assumption}

We model the occurrence time  of a trigger event, denoted by $\kappa$,
to be random, and the hazard rate is given by $l( t)$.
The survival probability at time $t$ is
\begin{equation}\label{su1}
P(t)=\exp\Big(-\int_0^tl(u)du\Big).
\end{equation}
The liquidation horizon is then defined by
\begin{equation}\label{stopping1}
\tau=\min\{T, \kappa\},
\end{equation}
where the constant $T\in(0,\infty)$ is a pre-determined time horizon. A direct computation yields the following proposition:
 \begin{proposition} \label{proposition2}
For $t<T$, the   density function of $\tau$  is
 $$
 f_{\tau}(t)= \displaystyle l(t)\exp\Big(-\int_0^tl(u)du\Big).
 $$
The probability that $\tau$ takes the value of $T$ is $\mathbb P(\tau=T)=P(T)$.
 \end{proposition}

Denote by $\mathcal G_t$ the event
$\{\tau>t\}=\{\mbox{the trigger event has {\it not}  occurred by the time $t$}\}$.
At any time $t<\tau$, i.e., the trigger event has not occurred prior to time $t$,  the agent's  objective is to find the optimal control for
\begin{equation}\label{op2}
\begin{array}{lll}
 \displaystyle \max_{\theta(\cdot)\in\Theta_t}\mathbb E\left[ \int_t^{\tau-} \Pi(\theta_u, X_u)du-\phi X_{\tau-}^2 \big|\mathcal F_t\lor \mathcal G_t\right]
\end{array}
\end{equation}
where
$$
\Pi(\theta_t, X_t)= g(\theta_t)\theta_t+f(\theta_t)X_t-\gamma \sigma^2 X_t^2,
$$
and $\mathcal F_t$ is the information structure available to the  agent  up to and including time $t$.  If  the trigger event occurs at time $t$, all market transactions will be    suspended at that time.   The agent will end up with  an  outstanding position $X_t$.
%Taking the conditional expectation with respect to the occurrence time $\tau$ (conditioning on the information set  $\mathcal G_t=\{\tau>t\}$), we obtain the optimal pre-event liquidation policy in case of uncertainty as the outcome of

It is worth noting  that
\begin{equation}\label{ko}
\begin{array}{lll}
\displaystyle E\Big[ \int_t^{\tau-} \Pi(\theta_u, X_u)du  \big|\mathcal F_t\lor \mathcal G_t\Big] &=&\displaystyle E\Big[ \int_{[t, T)} \mathbb I_{\{u< \tau\}} \Pi(\theta_u, X_u)du  \big|\mathcal F_t\lor \mathcal G_t\Big]\\
&=&\displaystyle E\Big[ \int_{[t, T)} \mathbb P( \tau> u|\mathcal G_t) \Pi(\theta_u, X_u)du  \big|\mathcal F_t\Big],
\end{array}
\end{equation}
and that
$$
 \mathbb P( \tau> u|\mathcal G_t)= \mathbb P( \tau> u|\tau> t)=e^{-\int_t^ul(r)dr}.
$$
Here, the indicator function $\mathbb I_{\{\cdot\}}$ takes the value 1 when its argument is true and the value $0$, otherwise.
 The last equality in Eq. (\ref{ko}) follows from the assumption that the trigger event is exogenous and does not depend on information structure $\mathcal F_t$.

Therefore,  we have
\begin{equation}\label{op212}
\begin{array}{lll}
&& \displaystyle \max_{\theta(\cdot)\in\Theta_t}\mathbb E\left[ \int_t^{\tau-} \Pi(\theta_u, X_u)du-\phi X_{\tau-}^2 \big|\mathcal F_t\lor \mathcal G_t\right]\\
&=&\displaystyle \max_{\theta(\cdot)\in\Theta_t} \mathbb E\Big[\int_{[t, T)} e^{-\int_t^ul(r)dr} \left[\Pi(\theta_u,  X_u)-\phi \cdot  l(u) X^2_u\right]du-\phi e^{-\int_{[t, T)}l(r)dr}X^2_{T-}\Big|\mathcal F_t\Big].
\end{array}
\end{equation}
That is, the optimal liquidating  problem with a random   horizon $\tau$ defined in
Eq. (\ref{stopping1}) is equivalent to an optimal liquidating  problem with a finite horizon
$T$,  a consumption process $\{\Pi(\theta_t,X_t)-\phi\cdot l(t) X_t^2\}_{t\ge0}$,
a discount process $\{l(t),t\ge0\}$,  and a  terminal condition $-\phi X_{T-}^2$.

\subsection{ Deterministic Control}

Let us first consider the case in which $\theta(\cdot)$ ranges only over the subclass
$\Theta^{det}_0$ of {\it deterministic} strategies in $\Theta_0$ satisfying the {\it hands-clean} condition (\ref{hand})\footnote{The {\it hands-clean} condition only makes sense for  the equivalent  problem  (\ref{op212}).
While  considering the original optimization  problem  (\ref{op2}), where the terminal time is a stopping time, the {\it hands-clean} condition is no longer valid.},  namely,
  $X_{T-}=0$. Thus, the agent's objective, before the trigger event occurs,    is to find the optimal control for
 $$
 \max_{\theta(\cdot)\in\Theta^{det}_0} \mathbb E\int_{[0, T)} e^{-\int_0^tl(u)du} \left[\Pi(\theta_t,  X_t)-\phi \cdot  l(t) X^2_t\right]dt.
 $$
 The cost function of the deterministic control problem is
 $$
\mathcal H(X_t,\theta_t,\Lambda_t,t)\equiv P(t)\big[\Pi(\theta_t,X_t)-\phi\cdot l(t)X_t^2\big]-\Lambda_t\theta_t,
$$
where $\Lambda_t$ is the {\it Lagrange multiplier}, and $P(t)$ is the survival probability defined in Eq. (\ref{su1}).
The differential equation for the deterministic system dynamics is
$$
\frac{dX_t}{dt}=-\theta_t \quad {\rm and} \quad  X_0=Q.
$$
  We assume that the Hamiltonian $\mathcal H$ has continuous first-order derivatives in the state, adjoint state, and  the control variable, namely,  $\{ X_t, \Lambda_t, \theta_t\}$.
Then the necessary conditions for having  an interior point optimum of the Hamiltonian
$\mathcal H$ at
$\{X_t^{det, *}, \Lambda_t^{det,*}, \theta_t^{det,*}\}$  are given by
\begin{equation}\label{kl0}
\left\{
\begin{array}{rll}
\displaystyle \frac{d X_t^{det,*}}{dt}&=&\displaystyle \frac{\partial \mathcal H}{\partial \Lambda}\bigg|_{(X_t^{det,*}, \theta_t^{det,*},\Lambda_t^{det,*},t)}\\
\displaystyle -\frac{d\Lambda_t^{det,*}}{dt}&=&\displaystyle \frac{\partial \mathcal H}{\partial X}\bigg|_{(X_t^{det,*}, \theta_t^{det,*},\Lambda_t^{det,*},t)}\\
\displaystyle 0&=&\displaystyle \frac{\partial \mathcal H}{\partial \theta}\bigg|_{(X_t^{det,*}, \theta_t^{det,*},\Lambda_t^{det,*},t)}.
\end{array}
\right.
\end{equation}
It follows from   Eq. (\ref{kl0}) that
\begin{equation}\label{er}
\left\{
\begin{array}{l}
\displaystyle   \ddot{X}_t^{det,*}=  l(t)\dot{X}_t^{det,*}+\frac{\gamma \sigma^2+(\phi-\frac{\eta}{2}) \cdot l(t)}{\nu}X_t^{det,*} \\
X_0^{det,*}=Q\\
X_T^{det,*}=0.
\end{array}
\right.
\end{equation}
Regarding  this linear second-order  boundary value problem (BVP), its existence and uniqueness are  standard. Interested readers can refer to, for example Hwang \cite{Hwang03}, for more details.
%  \begin{theorem}
%  Consider the two-point BVPs of the form:
%  \begin{equation}\label{BVP}
%  \left\{
%  \begin{array}{l}
%  \displaystyle y^{\prime\prime}(t)=f(t, y(t),y^\prime(t))\\
%  y(a)=\alpha,\quad y(b)=\beta.
%  \end{array}
%  \right.
%  \end{equation}
%Suppose $f$ is continuous on the domain
%$$
%D=\{(t,y,z):a\le t\le b, -\infty<y<\infty,-\infty<z<\infty\}
%$$
%and the partial derivatives $f_y$ and $f_z$ are also continuous on $D$.
%If
%  \begin{enumerate}
%\item $f_y(t,y,z)>0$ for all $(t,y,z)\in D$,
%\item there exists a constant $M$ such that $|f_z(t,y,z)|\le M$, for all $(t,y,z)\in D$,
%\end{enumerate}
%then the BVP (\ref{BVP}) has a unique solution.
%  \end{theorem}
% \begin{proof}

%  \end{proof}
%
Consider the case  when   $l(t)\equiv \lambda\ne 0$, which corresponds to the case of  constant hazard rate, an  explicit solution is given by
$$
\left\{
\begin{array}{l}
\displaystyle  X_t^{det,*}=Qe^{\frac{\lambda}{2}t}\frac{\sinh(\alpha(T-t))}{\sinh(\alpha T)}, \\
 \displaystyle  \theta_t^{det,*}=-Qe^{\frac{\lambda}{2}t}\frac{\left[\frac{\lambda}{2}\sinh(\alpha(T-t))-\alpha\cosh(\alpha(T-t) )\right]}{\sinh(\alpha T)},
  \end{array}
\right.
$$
where
$$
\alpha=\sqrt{\frac{\lambda^2}{4}+\frac{\gamma\sigma^2+(\phi-\frac{\eta}{2})\lambda}{\nu}}.
$$
\quad\\

It is worth noting that
(i) when $\lambda=0$, the model degenerates to Model 1;
(ii) as $\phi\to\infty$, $\lim_{ \phi\to \infty} \theta_t^{det, *}=0$ and  $\lim_{ \phi\to \infty} X_t^{det, *}=0$,\ for all $t\in(0,T]$; and $\lim_{ \phi\to \infty} \theta_0^{det, *}=\infty$. That is, as the final liquidation fee, $\phi$ per share, approaches infinity, the trader would  immediately complete the transaction  at the beginning of the trading horizon.

\subsection{Dynamic Programming Approach}

Let us consider the case of allowing    dynamic updating, i.e., replacing $\Theta^{det}_0$ by the entire class of admissible strategies  $\Theta_0$.
Let $F(t,q)$ denote the optimal value function of Eq. (\ref{op212}) at any time  prior to the occurrence of the trigger event.  Under appropriate  regularity  assumptions, $F$   satisfies the following HJB equation:
 \begin{equation}\label{model2}
\displaystyle  l(t)  F=\partial_t F -[\gamma\sigma^2+\phi\cdot l(t)]q^2- \min_{\theta_t\in\Theta_t}\Big\{\nu\theta^2_t+(\partial_qF+\eta q)\cdot \theta_t\Big\}
 \end{equation}
subject to the terminal condition: $F(T,q)=-\phi q^2$. Here,  $l(t)$ is the given hazard rate.
Similarly,  we  consider relaxing the constraints associated with the HJB equation
and solve the unconstrained optimization problem.
We then prove that the obtained optimal control does satisfy all the constraints.
The associated optimal trading  strategy is
$$
\theta_{t}^{\phi, *}=-\frac{1}{2\nu}\big(\partial_{q}F+\eta q \big),
$$
and hence the value function satisfies
\begin{equation}\label{01}
\left\{
\begin{array}{lll}
 \displaystyle \partial_t F      -[\gamma\sigma^2+\phi\cdot l(t)]q^2  +\frac{1}{4\nu}  \big(\partial_{q}F+\eta q \big)^2-l(t) F=0 \\
 F(T-,q)=-\phi q^2.
\end{array}
\right.
\end{equation}
Regarding Eq. (\ref{01}), we have  the following theorem for  the uniqueness of  classical  solutions.

\begin{theorem}\label{theorem3}
There is at most one  $\mathcal C^{1,2}([0,T)\times(0,\infty))$   solution to  Eq. (\ref{01}).
\end{theorem}
\begin{proof}
Suppose $f_1$ and $f_2$ are  two $\mathcal C^{1,2}([0,T)\times(0,\infty))$  solutions to Eq. (\ref{01}). Define  $\widetilde {f}=f_1-f_2$.  Then the new function $\widetilde{f}$ satisfies  the following problem:
$$
\left\{
\begin{array}{lll}
 \displaystyle \partial_t \widetilde{f}
+\frac{1}{4\nu}\big[ \partial_{q}(f_1+f_2) +2\eta q \big]\partial_q \widetilde{f}   -l(t)\widetilde{f} =0\\
 \displaystyle \widetilde{f} (T-, q)=0.
\end{array}
\right.
$$
Since the evolution equation for $\widetilde {f}$ is linear and first-order, one can solve the above problem  explicitly by the method of characteristics, and find that $\widetilde{f}\equiv 0$ is the unique solution to this problem. As a result, $f_1\equiv f_2$.
\end{proof}
\quad\\

Similar to Section 3.2, we consider an ansatz that is  quadratic in  the variable $q$:
$$
F(t,q)=\widetilde{a}(t)+\widetilde{ b}(t)q +\widetilde{c}(t)q^2.
$$
Substituting the ansatz into
  Eq. (\ref{01}),  we know that the coefficients   $\widetilde{a}(t)$,  $\widetilde{b}(t)$ and $\widetilde{c}(t)$  must satisfy  the following partially decoupled  system:
\begin{equation}\label{ep2}
\left\{
\begin{array}{l}
 \displaystyle \dot{\widetilde c}(t)=l(t)\widetilde c(t)+\gamma\sigma^2+\phi\cdot l(t)-\frac{1}{4\nu}[2\widetilde c(t)+\eta]^2\\
 \displaystyle   \dot{\widetilde b}(t)=l(t)\widetilde b(t) -\frac{1}{2\nu}\widetilde b(t)[2\widetilde c(t)+\eta]\\
 \displaystyle \dot{\widetilde a}(t)=l(t)\widetilde a(t)-\frac{1}{4\nu}\widetilde b^2(t)
\end{array}
\right.
 \end{equation}
with  terminal conditions: $\widetilde a(T-)=0$, $\widetilde b(T-)=0$ and $\widetilde c(T-)=-\phi$.

It is straightforward to verify that
$$
\widetilde{b}(t)\equiv 0\quad  {\rm and}\quad  \widetilde{a}(t)\equiv 0.
$$
However, the equation satisfied by  $\widetilde c(t)$
is a Riccati equation, which can be reduced to a second-order linear ODE:
  \begin{equation}\label{second}
  u^{\prime\prime}-l(t)u^\prime-\frac{\gamma\sigma^2+(\phi-\frac{\eta}{2})l(t)}{\nu}u=0,
  \end{equation}
where $u$ is defined implicitly via  $\widetilde c(t)=\frac{\nu u^\prime}{u}-\frac{\eta}{2}$.
For this  second-order linear  ODE, its existence and uniqueness are  standard.
Even though we know the existence and uniqueness of the solution,
it is still difficult to solve it in a closed-form for a general hazard rate $l(t)$.
%The closed-form solution can only be obtained when there exist certain relations between the coefficients.
The above   second-order linear ODE  can be easily solved in two cases: (i) its coefficients are constant;
or (ii) its coefficients adopt particular forms.

If closed-form solutions cannot be obtained, finite difference method can be applied
to solving the   BVP numerically.
For more details, see, for example, Hwang \cite{Hwang03}.

\begin{theorem}\label{theorem4}
{\bf(Constant hazard rate).} When  the hazard rate is a constant, i.e., $l(t)\equiv \lambda$, the unknown function $\widetilde c(t)$ can be explicitly solved.
It is given by
$$
\widetilde c(t)=\frac{1}{2\hat\xi}\left[\frac{\hat\zeta e^{-2 \alpha(T-t)}-1}{\hat \zeta e^{-2 \alpha(T-t)}+1}\right]+\frac{\lambda\nu-\eta}{2}
$$
with
$$
\alpha=\sqrt{\frac{\lambda^2}{4}+\frac{\gamma\sigma^2+(\phi-\frac{\eta}{2})\lambda}{\nu}},
\quad
\hat\zeta=\frac{1-\hat\xi(2\phi+\lambda\nu-\eta)}{1+\hat\xi(2\phi+\lambda\nu-\eta)}
\quad {\rm and} \quad \hat\xi=\frac{1}{2 \alpha\nu}.
$$
\end{theorem}
\quad\\

A direct verification yields  that when $\lambda=0$, $\widetilde c(t)=c(t)$.
The results derived under Model 2  coincide with  those  derived under Model 1.
The optimal liquidating strategy  for the unconstrained problem can then be derived through the following relation:
\begin{equation}\label{opl14}
\theta_t^{\phi, *}=-\frac{1}{2\nu}[2\widetilde c(t)+\eta]q,
\end{equation}
and hence,
$$
X^{\phi, *}_t=  Q\cdot \exp\left(\int_0^t\frac{1}{2\nu}[2\widetilde c(u)+\eta]du\right).
$$
The following theorem provides us the relation between the optimal deterministic control and  the optimal stochastic control.

\begin{theorem}\label{theorem5}
When the transaction fees  involved by liquidating the outstanding position $X_{T-}$  approaches to infinity, the optimal stochastic control  process $(\theta^{\phi,*}_t)_{t\in[0,T)}$ before the trigger event occurs    converges (point-wise)  to the optimal deterministic control process $(\theta^{det,*}_t)_{t\in[0,T)}$. Meanwhile, the optimal trajectory $X_t^{\phi, *}$ converges (point-wise) to the one determined in the deterministic system $X_t^{det,*}$: as  $\phi\to \infty$,  for any time $t\in[0,T)$,
\begin{enumerate}
\item $\displaystyle \lim_{\phi\to \infty} \theta^{\phi,*}_t= \lim_{\phi \to\infty} \theta^{det,*}_t$;
\item $\displaystyle \lim_{\phi\to \infty}X_t^{\phi,*}= \lim_{\phi\to\infty} X_t^{det,*}$.
\end{enumerate}

\end{theorem}

\begin{proof}
The proof of this theorem  is very similar to that of Theorem 2.
Therefore we will not provide all  the details; instead we will
just outline  the proof as follows.
First, a simple calculation yields that
 $$
 \begin{array}{lll}
\displaystyle \lim_{\phi\to\infty} X_t^{\phi,*}&=&\displaystyle  \lim_{\phi\to\infty} \displaystyle Q e^{-(\alpha-\frac{\lambda}{2})t}\frac{\hat\zeta e^{-2\alpha(T-t)}+1}{\hat\zeta e^{-2\alpha T}+1}\\
&=&\displaystyle \lim_{\phi\to \infty}Q e^{-(\alpha-\frac{\lambda}{2})t}\frac{  e^{-2\alpha(T-t)}-1}{  e^{-2\alpha T}-1}\\
&=&\displaystyle \lim_{\phi\to\infty}X_t^{det,*}.
\end{array}
 $$
Following the relations
 $$
\left\{
 \begin{array}{lll}
\displaystyle  \theta^{\phi, *}_t=-\frac{1}{2\nu} [2\widetilde c(t)+\eta]X_t^{\phi,*},\\
\displaystyle \lim_{\phi\to\infty} \frac{1}{2\nu}[2\widetilde{c}(t)+\eta]=\lim_{\phi\to\infty} \Big[\alpha\frac{e^{-2\alpha(T-t)}+1}{e^{-2\alpha(T-t)}-1}+\frac{\lambda}{2}\Big],
\end{array}
\right.
$$
we can further verify that  $\lim_{\phi\to\infty} \theta^{\phi,*}_t=\lim_{\phi\to\infty}\theta^{det,*}_t$.
\end{proof}
\quad\\

We remark that if condition (\ref{cond1}) is satisfied, then
$$
\frac{d\widetilde c(t)}{dt}<0\quad{\rm and}\quad 2\widetilde c(0)+\eta<\lambda\nu-\frac{1}{\hat\xi}<0.
$$
 Therefore,   $2\widetilde c(t)+\eta<0$ always holds for any time $t\in[0,T)$.
We can further verify that
(a1) $\theta^{\phi,*}_t\ge0$ holds for any time $t\in[0,T)$; and
(a2) $\int_{[0, T)}\theta^{\phi,*}_tdt\le Q$.
That is, the obtained optimal trading strategy in Eq. (\ref{opl14}) is also the optimal trading strategy for the constrained problem.

\subsection{Numerical Results}

In this section, we provide some numerical results to illustrate the effects of exogenous  trigger event on the agent's  liquidating strategy.
Suppose the size of the target order to be liquidated is $Q=100\; units$, the liquidation time $T=1\; day$, and the hazard rate at which the trigger event occurs is $\lambda=1$.
The model parameters' values are set as follows:
$$
\gamma=0.1,  \;\; \sigma=0.2,\;\; \eta=0.001,\;\; \nu=0.003,\;\;  \phi=0.1.
$$

Figure \ref{model12} provides a comparison of liquidation strategies under  two different settings: one without  trigger event (Model 1), and the other with
trigger event (Model 2).
In the upper-panel and middle-panel plots given in Figure \ref{model12},
we can observe that  an exogenous trigger event occurs at time $t=0.46$.
At that time, all trades are suspended in Model 2:
$$
\theta_t|_{t\in (0.46,1]}=0\quad \mbox{and}\quad X_t|_{t\in (0.46,1]}=X_{t=0.46}.
$$
Since our    objective is to liquidate a large position before time $T=1$ (Model 1) or time $\tau=\min\{0.46, 1\}$ (Model 2), agents facing the scenario that an exogenous trigger event  might  occur during the trading horizon (Model 2) would like to  accelerate the rate of liquidating
to reduce their exposure to potential position risk and eventually in a smaller position when the trigger event occurs.
Their strategy  is more ``convex'' compared with those who are not threatened by this trigger event, as can be seen from the upper-panel plot given in  Figure \ref{model12}.

\begin{figure}[H]
\begin{center}
\includegraphics[width=5.5in, height=4.3in]{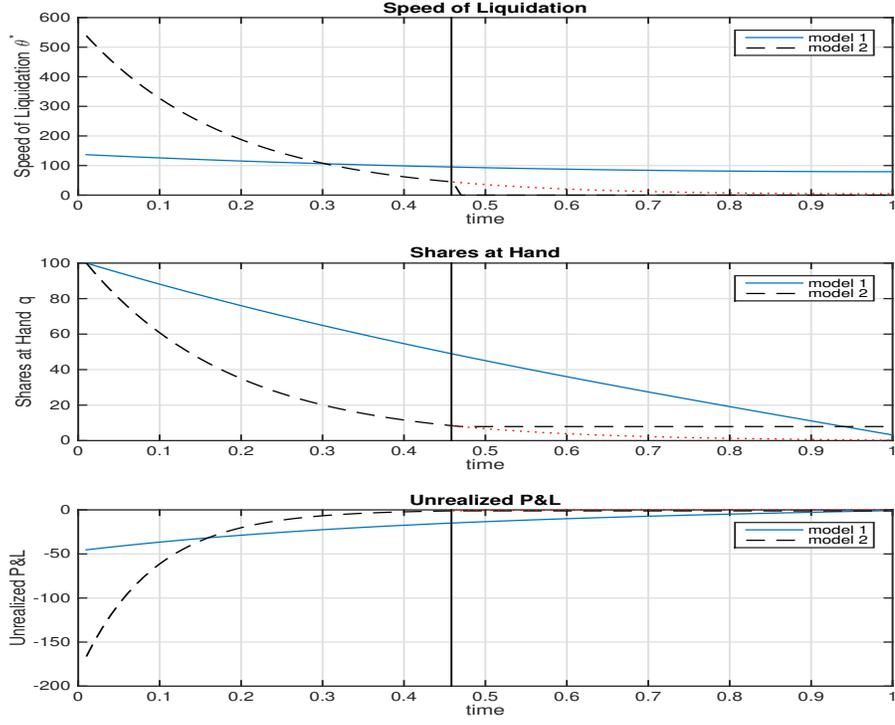}
\caption{ A comparison of liquidation strategies under two different settings:
one without  trigger event (Model 1), and the other with trigger event (Model 2).}
\label{model12}
\end{center}
\end{figure}

The lower-panel plot given in Figure \ref{model12} depicts the updated unrealized Profit  and Loss (P\&L) profile of the  DP problems as a function of time $t$:
$$
 U(t,q)= c(t)q^2 \quad {\rm and} \quad
 F(t,q)= \widetilde c(t)q^2.
$$
Notice that at any time $t\in[0,T]$, according to the DP principle,
the value function at time $t=0$ can be written as follows:
$$
U(0,Q)= \underbrace{ R^*_t-\gamma [V^*,V^*]([0,t))}_{\mbox{realized P\&L}} + \underbrace{U(t,X_t)}_{\mbox{unrealized P\&L}}.
$$
Here $R^*_t-\gamma [V^*,V^*]([0,t))$ can be regarded as the realized
 P\&L, and $U(t,X_t)$ can be regarded as the unrealized P\&L.
As we can see from Figure \ref{model12}, at the very beginning,
due to the potential position risk incurred by  exogenous trigger events,
the unrealized  P\&L under the second setting is significantly smaller
than that under the first setting.
This gap would eventually be narrowed through the adjustment of  the trading strategy,
and at time $t=0.15$, before the occurrence of the trigger event,  this situation is completely reversed.

\begin{figure}[H]
\begin{center}
\includegraphics[width=5.5in, height=3.5in]{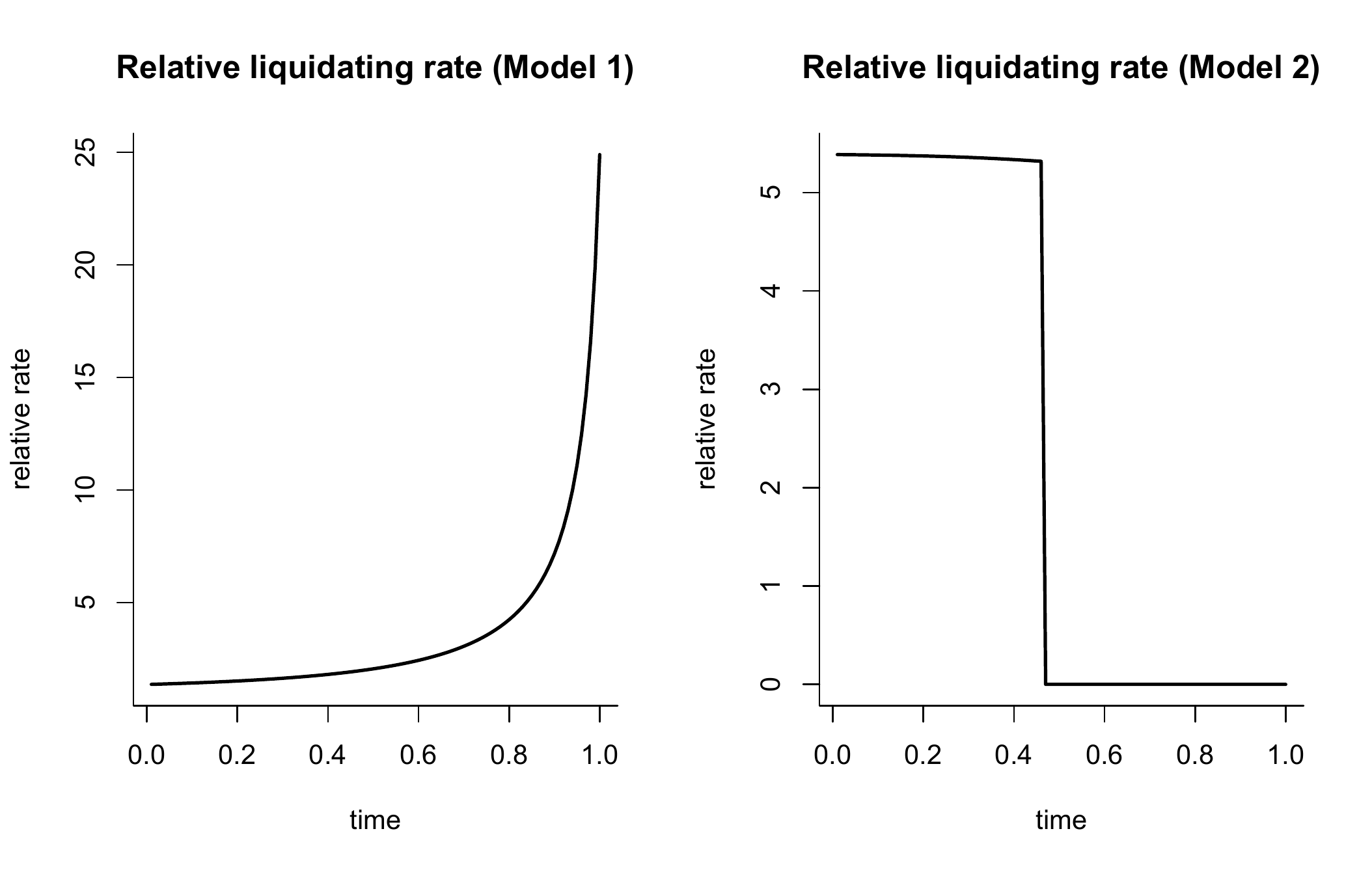}
\caption{ A comparison of  relative liquidating speed $\theta/q$   under two different settings:
one without  trigger event (Model 1), and the other with trigger event (Model 2).}
\label{bty}
\end{center}
\end{figure}

Figure \ref{bty} displays the relative liquidating speeds $\theta_t/X_t$ under the two different  settings.
We clearly see that the relative liquidating rate  under the first setting (Model 1) depends on
time-to-maturity in a monotonic way.
Indeed, as  $t$ approaches to the time horizon  $T$, there is a real need for an agent  to liquidate
because the liquidation cost $\phi X_{T-}^2$ at time $T$  is high.
However,  if the agent faces an additional risk that an exogenous trigger event might occur during the liquidating horizon, when the related risk is high, he/she needs to trade faster to reduce this risk.
As in the above figure, the relative liquidating rate under the second setting  shows a nearly   constant  during the  period  $[0,\tau)$.

\section{Optimal Liquidation Strategy Subject to Counterparty Risk  (Model 3) }
	
In this section, we assume that  the trigger event is not exogenous.
It is incurred  by the  evolvement of the market value of the stock issuer.

\subsection{The Hitting Time}
Let the stock issuer's market value $Y_t$ evolves over time  according to
 $$
\left\{
\begin{array}{lll}
\displaystyle dY_t=Y_t\big(\beta  dt+\xi dW^Y_t\big),\quad Y_0=y_0\\
\displaystyle dW^Y_t d W^S_t= \rho,
\end{array}
\right.
$$
where the constant  $\beta$ is the mean rate of return of the company, the constant $\xi$ is the volatility,   $\{W^S_t\}$ and $\{W^Y_t\}$ are two correlated  Brownian motions,  and the constant $\rho$ is the correlation coefficient, with $|\rho|<1$.
Thus, we have
$$
Y_t=y_0e^{(\beta-\frac{\xi^2}{2})t+\xi W_t^Y}.
$$
We assume that once the company's  market  value falls down to  a pre-determined limit $\alpha^*>0$
(it is pre-assumed that $y_0> \alpha^*$),  a great switch will be involved in this company and the liquidation process will be forced to suspend.
Let
$$
m(y_0)=\frac{1}{\xi}\ln\left(\frac{y_0}{\alpha^*}\right)>0
\quad {\rm and} \quad
\alpha=\frac{1}{\xi}\left(\frac{\xi^2}{2}-\beta\right),
$$
and define $ \widehat{W}^Y_t=W^Y_t+\alpha t$.
According to the information, the time at which this switch   occurs is  defined by\footnote{We use the fact that $-W_t^Y$ has the same distribution as $W_t^Y$.}
$$
\kappa_{m(y_0)}= \inf\{t\ge0: \widehat{W}^Y_t= m(y_0)\}.
$$
The liquidation horizon is then given by $\tau_{m(y_0)}=\min\{\kappa_{m(y_0)}, T\}$,
which is a stopping time.
Due to the Markov property of the  Brownian motion  $\{\widehat W^Y_t\}$, given $\kappa_{m(y_0)}>t$ and $Y_t=y$, the conditional distribution of $\kappa_{m(y_0)}$ is given by
$$
\kappa_{m(y_0)}|\{\kappa_{m(y_0)}>t\lor Y_t=y\}=t+\kappa_{m(y)}.
$$
Therefore, we obtain
\begin{equation}\label{stop}
\tau_{m(y_0)}|\{\tau_{m(y_0)}>t\lor Y_t=y\}=t+\min\{\kappa_{m(y)}, T-t\}=:\tau_{t, m(y)}.
\end{equation}

\begin{proposition}\label{proposition3}
For the Brownian motion,
$\widehat{W}^Y_t = W^Y_t+\alpha t, \ \alpha\ne 0$, define
$$
\kappa_{m}= \inf\{t\ge0: \widehat{W}^Y_t= m\}, \ \ m>0.
$$
The Laplace transform of $\kappa_m$ is
$$
\mathbb E[e^{-u\kappa_m}]=e^{\alpha m-m\sqrt{2u+\alpha^2}}, \quad \mbox{for all $u>0$}.
$$
\end{proposition}

%\begin{proof}
%% While the details of the proof will not be given here, the intuitive content is clear.  
% Let $\sigma $ be a positive number and set
% $$
% Z(t)=\exp\left\{ \sigma \widehat W^Y_t-\left(\sigma\alpha+\frac{1}{2}\sigma^2\right)t\right\}
% $$
%and  $\{Z(t)\}$ is a martingale.
%By optional sampling theorem,  we have
% $$
% \mathbb E\left[\exp\left\{ \sigma \widehat W^Y_{t\land\kappa_m}-\left(\sigma\alpha+\frac{1}{2}\sigma^2\right)(t\land\kappa_m)\right\}\right]=1, \quad t%\ge0.
% $$
% It can be shown through the derivation that
% $$
%  \mathbb E\left[\exp\left\{ \sigma m-\left(\sigma\alpha+\frac{1}{2}\sigma^2\right) \kappa_m\right\}\mathbb I_{\{\kappa_m<\infty\}}\right]=1,
%  $$
%  and  that
%  $\mathbb P\{\kappa_m<\infty\}\le 1$.  Using these facts, one can deduce  the Laplace transform $\mathbb E[e^{-u\kappa_m}]$.
%For more details, we refer readers to Steven \cite{SES}.
%\end{proof}

%The random variable $\kappa_{m}$ is a {\it stopping time} because it chooses its values based on the path of the asset price up to time $\kappa_{m}$. And hence, the stopping time truncated at $T$, i.e., $\tau_{m}=\min\{\kappa_{m}, T\}$, is also a stopping time.

Recall that, at any time $t$  prior to the time horizon $\tau_{t, m(y)}$
(which is defined in Eq. (\ref{stop})) with initial value $Y_t=y$ and $X_t=q$,  the agent's objective is to find the optimal control for
\begin{equation}\label{problem21}
\begin{array}{lll}
 H(t, y,q)&=&\displaystyle \max_{\theta(\cdot)\in\Theta_t}\mathbb E\left[\int_t^{\tau_{t, m(y)}-}\Pi(\theta_r, X_r)dr-\phi X_{\tau_{t, m(y)}-}^2\Big|\mathcal F_t \right]\\
 &=& \displaystyle \max_{\theta(\cdot)\in\Theta_t}\mathbb E\left[\int_{[t, T)}\mathbb I_{\{r<\tau_{t, m(y)}\}}\Pi(\theta_r, X_r)dr-\phi X_{\tau_{t, m(y)}-}^2\Big|\mathcal F_t \right],
\end{array}
  \end{equation}
  where
  $$
\Pi(\theta_t, X_t)= g(\theta_t)\theta_t+f(\theta_t)X_t-\gamma \sigma^2 X_t^2.
$$
%Different from Model 2, here $\tau_{t,m(y)}$ is $\mathcal F_t$-adapted,
%depending  on the market value of the stock issuer.
%The computation of the right-hand side of Eq. (\ref{problem21}) requires that we know something about the dependence between the randomly-terminated horizon %$\tau_{t,m(y)}$ and the ``consumption'' function $\Pi(\theta_t, X_t)$.
%This, however, can be difficult to quantify.
%In this situation, we  provide a DP approach to Problem (\ref{problem21}).
We verify comparison principles for viscosity solutions and  characterize the value function as the unique viscosity solution of the associated HJB equation.
  %These make  it almost impossible to   solve the liquidation problem   (\ref{problem21}) using the same method as we did in Model 2.

\begin{remark}
Assume that the company's market value has not hit   the pre-determined level $\alpha^*$ by time $t$. By the definition of $m(y)$,  we have
$m(y)\to \infty$ as  $y\to \infty$.
 According to Proposition 2,   $\lim_{y\to\infty}\mathbb E[e^{-u\kappa_{m(y)}}]=0$, for all $u>0$. This implies that,
$$
\begin{array}{lll}
\displaystyle 0\le  e^{-uN}\mathbb P(\kappa_{m(y)}\le N)\le \mathbb E[e^{-u\kappa_{m(y)}}],
\end{array}
$$
 for any positive integer $N$.
Hence,  $\lim_{y\to\infty}\mathbb P(\kappa_{m(y)}\le N)=0$.
Passing to the limit   $N\to\infty$, we conclude that
 $$
 \lim_{y\to \infty} \mathbb P(\kappa_{ m(y)}< \infty)=0\quad \mbox{and}\quad \lim_{y\to \infty} \mathbb P(\kappa_{ m(y)}= \infty)=1.
 $$
Therefore, $\lim_{y\to \infty} \mathbb P(\tau_{t, m(y)}= T)=1$.%, and hence,
%the optimization problem that  we consider here becomes the one that  we considered in Model 1.
\end{remark}

\subsection{Dynamic Programming Approach }

In this section, we  discuss some analytical properties of the value function $H$ without proofs.
Some technical proofs will be  provided in Sections 5.3 and 5.4.

\begin{theorem}\label{theorem6}
Let
 $H(t,  y, q)$  denote  the value function in Eq. (\ref{problem21}) at any time $t$ before the process $\{Y_t, t\ge 0\}$ touches the pre-determined limit  $\alpha^*$ and $(Y_t,X_t)=(y,q)$.
Suppose the value function $H(t,y,q)$ is sufficiently smooth\footnote{$\mathcal C^{1,2,2}([0,T)\times (\alpha^*,\infty)\times (0,\infty))$ is the space of functions $f(t, y,q)$ which is continuously differentiable in $t$, and twice continuously differentiable in $y$ and $q$.}, i.e.,
$H\in \mathcal C^{1,2,2}([0,T)\times  (\alpha^*,+\infty)\times(0,+\infty) )$,
then $H(t,y,q)$ satisfies the HJB equation
 \begin{equation}\label{HJB2}
 \displaystyle -(\partial_t+\beta y\partial_y+\frac{1}{2}\xi^2y^2\partial_{yy})  H  + \gamma  \sigma^2 q^2 -  \max_{\theta_t\in\Theta_t}\Big\{ g(\theta_t)\theta_t+ f(\theta_t) q-\partial_{q}H\cdot \theta_t\Big\}=0,
\end{equation}
in the region
$$
\{(t, y,q): 0\le t<T, y> \alpha^*,  q>0\}
$$
and satisfies the boundary conditions
\begin{description}
\item[$(a)$] $\displaystyle H(T, y, q)=-\phi q^2,\quad\quad\quad y>\alpha^*$,  \ \quad\quad \quad {\rm (35.a)}
\item[$(b)$] $\displaystyle H(t, \alpha^*, q)=-\phi q^2,\quad\quad 0\le t< T$,  \quad \quad\quad {\rm (35.b)}
\item[$(c)$]  $\displaystyle \lim_{y\to \infty} H(t,y,q)=U(t,q)$,\quad \quad\quad \quad \quad \quad\quad \quad {\rm (35.c)}
\end{description}
where   $U(t,q)$ is the value function of the optimization problem that we  considered in  Model 1.
\end{theorem}

%The  boundary condition  ($b$) above follows from the fact that when the geometric Brownian motion $\{Y_t\}$ hits the  level $\alpha^*$, it immediately falls below $\alpha^*$.
%In fact, because it has nonzero quadratic variation, the asset value $Y_t$ oscillates, rising and falling across the level $\alpha^*$ infinitely many times immediately after hitting it.
%The liquidating fees involved is $\phi X^2_{t}$ when the asset value hits $\alpha^*$
%at time $t$ because the liquidating process is on the verge of suspending.
%The only exception to this is if the level $\alpha^*$ is first reached at the time horizon $T$, for then there is no time left for the suspending.
%In this case, the value function is given by the terminal condition ($a$).
%In particular, the function $H(t, y,q)$ discussed in this section  is continuous at the corner of its domain where $t=T$ and $y=\alpha^*$.
%The natural condition ($c$) directly follows from Remark 1.

It is worth noting  that  for any $y\ge \alpha^*$,  $ H(t,y,q)\le U(t,q)$, and for any
$y_1\ge y_2\ge \alpha^*$,  $ H(t, y_1, q)\ge H(t, y_2, q)$.
The rationale behind this is intuitive.
Compared with ``default-free'' model (Model 1), greater counterparty risk gives a smaller value function.

\subsection{Monotonicity and   Continuity of $H(t,y,q)$ }

In  Section 5.2, we presented without proof the primary analytical properties of the value function $H(t,y,q)$.
In this section, we prove the monotonicity, growth rate control and  continuity of $H(t,y,q)$ as follows.

\begin{theorem} \label{theorem7}
Assume that condition (\ref{cond1}) is satisfied:
$2\phi>\eta+2\sigma\sqrt{\gamma\nu}$.
Then, we have
\begin{description}
\item[(i)] {\bf(Monotonicity)}
$H(t,y,q)$ is an increasing function in $y$, and a decreasing function in $t$;
\item[(ii)]{\bf(Continuity)}
 $H(t,y,q)$ is locally H\"{o}lder continuous  in $t$ with exponent  $1/2$, and locally Lipschitz continuous in both $y$ and $q$   in  $[0,T)\times (\alpha^*,\infty)\times (0,\infty)$;
 \item[(iii)] {\bf(Growth Rate Control)} $H(t,y,q)$ satisfies a quadratic  growth condition with respect to the inventory variable $q$: for any $(t, y, q)\in [0,T)\times (\alpha^*, \infty)\times(0, \infty)$,
$$
H(t, y, q)\le \left(\phi + \big[\frac{(2\phi-\eta)^2}{4\nu}-\gamma\sigma^2\big](T-t)\right) q^2.
$$
%where $U(t,q)$ is the value function of the  associated optimization problem without counterparty risk, that is, Model 1.

 \end{description}
\end{theorem}

To avoid confusion, in the sections to follow,
we denote $X(t)$, the number of shares at time $t$, and
$\{X_{t,q}^{\theta}(u)\}_{u\ge t}$,  a trajectory of $X(\cdot)$  given $ X(t)=q$ and a trading strategy $\theta$.
To prove Theorem \ref{theorem7},
 we first convert the original control problem into a problem without terminal bequest function.  Since $g(x)=-\phi x^2$ is continuously differentiable, and $\mathbb E[\tau_{t,m(y)}|\mathcal F_t ]< T<\infty$, we can apply Dynkin's formula  to    $-\phi X^2(t)$ and rewrite the value  function $H$ as
$$
  H(t,y,q)=-\phi q^2+\max_{\theta(\cdot)\in\Theta_t}\mathbb E\left[\int_t^{\tau_{t, m(y)}-}L( \theta_r,X_{t,q}^{\theta}(r))dr \Big|\mathcal F_t\right]
 $$
where
 \begin{equation}\label{L}
  L(\theta,q)=\Pi(\theta, q)+2\phi q\theta.
 \end{equation}
 Define a new value function as
 $$
 \widehat{H}(t,y,q)=\max_{\theta(\cdot)\in\Theta_t}\mathbb E\left[\int_t^{\tau_{t,m(y)}-}L(  \theta_r, X_{t,q}^{\theta}(r))dr \Big|\mathcal F_t\right].
 $$

\noindent
{\bf Proof of Theorem 7.}

 \begin{description}
\item[(i)]  One approach to  verify the monotonicity in $y$  is  directly applying  the definition.    Let $\theta^{y_2,*}$ denote the optimal control process with respect to  the stopping time $\tau_{t, m( y_2)}$.  For any positive numbers $y_1\ge y_2> \alpha^*$, we have  $\tau_{ t, m(y_1)}\ge \tau_{t,m(y_2)}$. From this observation, we  have
$$
\begin{array}{lll}
 \widehat H(t,y_1,q)&=&\displaystyle \max_{\theta(\cdot)\in\Theta_t}\mathbb E\left[\int_t^{\tau_{t, m(y_1)}-}L(\theta_r, X^{\theta}_{t,q}(r))dr\Big|\mathcal F_t \right]\\
&\ge&\displaystyle   \mathbb E\left[\int_t^{\tau_{t, m(y_2)}-}L(\theta^{y_2,*}_r, X_{t, q}^{\theta^{y_2,*}})dr\big| \mathcal F_t \right]\\
&&+\displaystyle \max_{\theta(\cdot)\in\Theta_{\tau_{t, m(y_2)}}} \underbrace{\mathbb E  \left[\int_{\tau_{t, m(y_2)}}^{\tau_{t, m(y_1)}-}L(\theta_r, X_{\tau_{t,  m(y_2)}, X_{t,q}^{\theta^{y_2,*}}(\tau_{t,m(y_2)}-)}^{\theta}(r))dr\big|\mathcal F_t\right]}_{\mbox{(II)}}.\\
\end{array}
$$
Since
$$
\begin{array}{lll}
 L(\theta,q)=\Pi(\theta, q)+2\phi q\theta=-\nu\theta^2+( 2\phi-\eta) q\theta-\gamma\sigma^2 q^2.
 \end{array}
 $$
 Under the assumption that  $2\phi>\eta+2\sigma\sqrt{\gamma\nu}$, for any $q\in (0,+\infty)$,
 we can always choose  $\theta:=\frac{2\phi-\eta}{2\nu}q\ge0$  such  that $L(\theta, q)\ge0$, and hence,  (II) $\ge0$. Therefore,
$\widehat{H}(t,y_1,q)\ge  \widehat{H}(t,y_2,q)$, and hence  $H(t,y_1,q)\ge H(t,y_2,q)$.

Another  approach to this question  is to apply  the result in Section 3.2.  Let $U_T(t,q)$ denote the  value function of the optimization problem (\ref{e1})\footnote{The associated liquidation  problem without countparty risk.}   with time horizon $T$. Under condition (\ref{cond1}),
for any $T_1>T_2>t$,  we have Inequality (\ref{red}):
$$
U_{T_1}(t,q)>U_{T_2}(t,q).
$$
If we set $T_1=\tau_{t, m(y_1)}$ and $T_2=\tau_{t, m(y_2)}$, then
$$
H(t, y_1, q)=\mathbb E[U_{\tau_{t, m(y_1)}}(t, q)|\mathcal F_t]\ge\mathbb E[U_{\tau_{t, m(y_2)}}(t, q)|\mathcal F_t]=H(t, y_2,q).
$$

Similarly, we can verify   the monotonicity of $H(t, y,q)$ in $t$.   If we set   $\bar U(\iota, q)=U_T(t, q)$, where $\iota=T-t$ is the time   to  maturity, then, according to Proposition 1, for any $0\le \iota_2<\iota_1<T$,
$$
\bar U(\iota_1, q)>\bar U(\iota_2, q).
$$
For any  $0\le t_1<t_2< T$. Let $\iota_1=\tau_{t_1, m(y)}-t_1$ and $\iota_2=\tau_{t_2, m(y)}-t_2$. By the definition of $\tau_{t, m(y)}$, we have $\iota_1\ge \iota_2$, and hence
$$
H(t_1, y, q)=\mathbb E[\bar U(\iota_1, q)|\mathcal F_t]\ge \mathbb E[\bar U(\iota_2, q)|\mathcal F_t]=H(t_2, y, q).
$$

\item[(ii)] In order to prove the continuity of $ H$, it suffices to show that for  any two points $(t_1, y_1, q_1)$ and $(t_2, y_2, q_2)$ in the region
$$
\{(t, y,q): 0\le t<T, \alpha^* < y,  0 < q\},
$$
there exist three  $(t,y)$-independent, polynomial-growth (with respect to  $(q_1,q_2)$)   coefficients $K_1(q_1, q_2), K_2(q_1,q_2)$, and $K_3(q_1,q_2)$, such that
$$
\begin{array}{lll}
\displaystyle |\widehat H(t_1, y_1, q_1)-\widehat H(t_2, y_2, q_2)|
&\le &\displaystyle | \widehat H(t_1, y_1, q_1)-\widehat H(t_1, y_2, q_1)|\\
& &+|\widehat H(t_1, y_2, q_1)-\widehat H(t_1, y_2, q_2)|\\
&&\displaystyle +|\widehat H(t_1, y_2, q_2)-\widehat H(t_2, y_2, q_2)|\\
&\le&K_1 |y_1-y_2|+K_2|q_1-q_2|+K_3(|t_2-t_1|^{\frac{1}{2}}+|t_2-t_1|).
\end{array}
$$
%Here, $\omega^t$, $\omega^y$ and $\omega^q$ are some moduli of continuity such that
%$$
%\displaystyle  \lim_{|t_1-t_2|\to0} \omega^t(|t_1-t_2|)=0, \
%\displaystyle  \lim_{|y_1-y_2|\to0} \omega^y(|y_1-y_2|)=0,\ {\rm and} \
%\displaystyle  \lim_{|q_1-q_2|\to0} \omega^q(|q_1-q_2|)=0.
%$$
We divide the proof into three parts: one is for the  variable $y$, another  one is for the  variable $q$, and the rest  is for the variable $t$.\\

{\bf Step 1 (Variable  $y$).} For any positive numbers $y_1\ge y_2>\alpha^*$,    we have $\tau_{t, m(y_1)}\ge \tau_{t, m(y_2)}$, and hence
 \begin{equation}\label{proofA1}
 \begin{array}{lll}
&& \widehat{H}(t, y_1,q)\\
 &=&\displaystyle   \mathbb E\left[\int_t^{\tau_{t, m(y_2)}-}L(\theta^{y_1,*}_r, X_{t, q}^{\theta^{y_1,*}})dr+  \int_{\tau_{t, m(y_2)}}^{\tau_{t, m(y_1)}-}L(\theta^{y_1, *}_r, X_{\tau_{t,  m(y_2)}, X_{t,q}^{\theta^{y_1,*}}(\tau_{t,m(y_2)}-)}^{\theta^{y_1, *}}(r))dr\big|\mathcal F_t\right]\\

 &\le&\displaystyle  \widehat{H}(t, y_2, q)+\max_{\theta(\cdot)\in\Theta_{\tau_{t, m(y_2)}}}\mathbb E  \left[\int_{\tau_{t, m(y_2)}}^{\tau_{t, m(y_1)}-}L(\theta_r, X_{\tau_{t,  m(y_2)}, X_{t,q}^{\theta^{y_1,*}}(\tau_{t,m(y_2)}-)}^{\theta}(r))dr\big|\mathcal F_t\right],

% &\le&C_1(q) \sqrt{\mathbb E[\kappa_{m(y_1)}-\kappa_{m(y_2)}|\mathcal F_t]}\\
% &\le& K_1(q) \sqrt{\ln(y_1)-\ln(y_2)}\quad\quad\quad\quad\quad\quad\mbox{(Proposition 2)}\\
% &\le& \displaystyle K_1(q)\sqrt{\frac{1}{\alpha^*}|y_1-y_2|}\quad\quad\quad\quad\quad\quad\mbox{(Concavity of $\ln(\cdot)$)}\\
 \end{array}
 \end{equation}
 where $\theta^{y_1,*}$ is  the optimal control process with respect to  the stopping time $\tau_{t, m( y_1)}$.
 Thus, using part (i), we have
 \begin{equation}\label{kop}
 \begin{array}{lll}
 &&\displaystyle |\widehat{H}(t, y_1, q)-\widehat{H}(t, y_2, q)|=   \widehat{H}(t, y_1, q)-\widehat{H}(t, y_2, q)\\
 &\le& \displaystyle \max_{\theta(\cdot)\in\Theta_{\tau_{t, m(y_2)}}}\mathbb E  \left[\int_{\tau_{t, m(y_2)}}^{\tau_{t, m(y_1)}-}L(\theta_r, X_{\tau_{t,  m(y_2)}, X_{t,q}^{\theta^{y_1,*}}(\tau_{t,m(y_2)}-)}^{\theta}(r))dr\big|\mathcal F_t\right].
% &\le& \displaystyle \max_{\theta(\cdot)\in\Theta_{\tau_{t, m(y_2)}}}\mathbb E  \left[ \int_{\tau_{t, m(y_2)}}^{\tau_{t, m(y_1)}-}2\phi \theta_r X_{\tau_{t,  m(y_2)}, X_{t,q}^{\theta^{y_1,*}}(\tau_{t,m(y_2)}-)}^{\theta}(r)dr  \big |\mathcal F_t \right] \\
% &=&\displaystyle \max_{\theta(\cdot)\in\Theta_{\tau_{t, m(y_2)}}}\mathbb E  \left[ \int_{\tau_{t, m(y_1)}}^{\tau_{t, m(y_2)}-}2\phi  X_{\tau_{t,  m(y_2)}, X_{t,q}^{\theta^{y_1,*}}(\tau_{t,m(y_2)}-)}^{\theta}(r)d\dot { X}_{\tau_{t,  m(y_2)}, X_{t,q}^{\theta^{y_1,*}}(\tau_{t,m(y_2)}-)}^{\theta}(r)  \big |\mathcal F_t \right] \\
%
% &\le&\displaystyle \max_{\theta(\cdot)\in\Theta_{\tau_{t, m(y_2)}}}\left\{
% \phi \mathbb E\left[ (|X_{t,q}^{\theta^{y_1,*}}(\tau_{t,m(y_2)}-)|^2 -|X_{\tau_{t,  m(y_2)}, X_{t,q}^{\theta^{y_1,*}}(\tau_{t,m(y_2)}-)}^{\theta}(\tau_{t, m(y_1)})|^2) \big|\mathcal F_t \right]\right\}\\
% &\le& \displaystyle\max_{\theta(\cdot)\in\Theta_{\tau_{t, m(y_2)}}}\Big\{ 2\phi q\mathbb E\left[ \int_{\tau_{t, m(y_2)}}^{\tau_{t, m(y_1)}-} \theta_r dr\big| \mathcal F_t\right] \Big\}.
 \end{array}
 \end{equation}
 A completing square yields
 $$
 L(\theta, X)=-\nu\big[\theta-\frac{2\phi-\eta}{2\nu}X\big]^2+\big[\frac{(2\phi-\eta)^2}{4\nu}-\gamma\sigma^2\big] |X|^2,
 $$
 so\footnote{
  It is worth noting that condition (\ref{cond1}) implies that  $  \frac{(2\phi-\eta)^2}{4\nu}-\gamma\sigma^2>0$.
  }
$$
 \begin{array}{lll}
 &&\displaystyle   \displaystyle \max_{\theta(\cdot)\in\Theta_{\tau_{t, m(y_2)}}}\mathbb E  \left[\int_{\tau_{t, m(y_2)}}^{\tau_{t, m(y_1)}-}L(\theta_r, X_{\tau_{t,  m(y_2)}, X_{t,q}^{\theta^{y_1,*}}(\tau_{t,m(y_2)}-)}^{\theta}(r))dr\big|\mathcal F_t\right]\\
 & \le&\displaystyle  \big[\frac{(2\phi-\eta)^2}{4\nu}-\gamma\sigma^2\big] q^2\mathbb E\left[ \int_{\tau_{t, m(y_2)}}^{\tau_{t, m(y_1)}-}dr\big|\mathcal F_t\right]\\
 &=&\displaystyle  \big[\frac{(2\phi-\eta)^2}{4\nu}-\gamma\sigma^2\big] q^2\mathbb E\left[ \tau_{t, m(y_1)}-\tau_{t, m(y_2)}|\mathcal F_t\right]
 \end{array}
$$
 because
 $$
\left|X_{\tau_{t,  m(y_2)}, X_{t,q}^{\theta^{y_1,*}}(\tau_{t,m(y_2)}-)}^{\theta}(r)\right|^2\le \left|X_{t,q}^{\theta^{ y_1, *}}(\tau_{t,m(y_2)}-)\right|^2\le q^2.
 $$
% Noticing the facts  that
%   \begin{enumerate}
% \item $\displaystyle\mathbb E\left[\max_{0\le r\le T}|\theta_r|\right]<\infty$ \quad ($L_\infty$-integrability),
% \item $\displaystyle \int_t^T\theta_rdr\le q$\quad and \quad
%  $\displaystyle X_{t, q}^\theta(r)=q-\int_t^r\theta_udu\le q$,
%\end{enumerate}
%and that
%$$
%|L(\theta_r, X_{t,q}^\theta(r))|\le \displaystyle   \nu|\theta_r|^2+( 2\phi-\eta) q|\theta_r|+\gamma\sigma^2 q^2.
%$$
%Hence there exists a $(t,y)$-independent, polynomial-growth coefficient $C_1(q)$ so that
%$$
%\sqrt{\mathbb E\left[\int_t^T  \big|L(\theta_r, X_{t,q}^{\theta}(r))\big|^2dr \Big|\mathcal F_t \right]}\le C_1(q)<\infty.
%$$
By the definition of $\tau_{t, m(y)}$,  we have
$\tau_{t, m(y_1)}-\tau_{t, m(y_2)}\le \kappa_{m(y_1)}-\kappa_{m(y_2)}$.
According to Proposition 3 and the definition of $m(y)$, there exists a constant $c_0>0$ such that
$$
\mathbb E[\kappa_{m(y_1)}-\kappa_{m(y_2)}|\mathcal F_t]\le c_0 |\ln(y_1)-\ln(y_2)|,
$$
and hence,
\begin{equation}\label{hidden}
|\widehat{H}(t,y_1,q)-\widehat{H}(t,y_2,q)|\le c_0 \big[\frac{(2\phi-\eta)^2}{4\nu}-\gamma\sigma^2\big] q^2|\ln(y_1)-\ln(y_2)|.
\end{equation}

 Since $|\ln(y_1)-\ln(y_2)|\le \frac{1}{\alpha^*} |y_1-y_2|$, for any $y_1, y_2\in (\alpha^*, +\infty)$. There  exists   a  $(t,y)$-independent, quadratic-growth coefficient $K_1(q)$ so that
$$
 |\widehat{H}(t,y_1,q)-\widehat{H}(t,y_2,q)|\le K_1(q) |y_1-y_2|.
$$
\quad\\
 {\bf Step 2 (Variable $q$).}  Let $q_1, q_2\in (0,+\infty)$ satisfying $|q_1-q_2|\le 1$.
 Consider the value functions $\widehat{H}(t,y,q_1)$ and $\widehat{H}(t,y,q_2)$. By the definition and  the relation $|\max f-\max g|\le \max|f-g|$, we have
  \begin{equation}\label{proofA2}
 \begin{array}{lll}
 &&\displaystyle |\widehat{H}(t,y,q_1)-\widehat{H}(t,y,q_2)|\\
 &\le&\displaystyle \max_{\theta(\cdot)\in \Theta_t}\mathbb E\left[\int_t^T \mathbb I_{t\le r< \tau_{t, m(y)}}  \cdot \big|L(\theta_r, X_{t,q_1}^{\theta}(r))-L(\theta_r, X_{t,q_2}^{\theta}(r))\big|dr\Big|\mathcal F_t \right]\\
 &\le&\displaystyle \max_{\theta(\cdot)\in \Theta_t}\mathbb E\left[\int_t^T   \big|L(\theta_r,X_{t,q_1}^{\theta}(r))-L(\theta_r,X_{t,q_2}^{\theta}(r))\big|dr\Big|\mathcal F_t \right]\\
 &\le &\displaystyle K_2(q_1,q_2) |q_1-q_2|,
 \end{array}
 \end{equation}
  where $K_2(q_1,q_2)$ is a polynomial-growth  coefficient.  The last inequality follows from Definition 1 and   the fact that
  $$
  X_{t,q_1}^{\theta}(r)-X_{t,q_2}^{\theta}(r)=q_1-q_2,
  $$
  for any trading strategy $\theta\in\Theta_t$. \\

\noindent
{\bf Step 3 (Variable $t$).} Let $0\le t_1<t_2< T$, and $(y,q)\in (\alpha^*, \infty)\times (0,+\infty)$.
By the DP principle,
 $$
\begin{array}{lll}
\displaystyle  \widehat{H}(t_1,y,q)=\mathbb E\left[ \int_{t_1}^{t_2} L(\theta_r,X_{t_1,q}^{\theta^*}(r))dr+ \widehat{H}(t_2,Y_{t_1, y}(t_2), X_{t_1,q}^{\theta^*}(t_2))\Big|\mathcal F_{t_1}\right],
 \end{array}
 $$
 where $\theta^*$ is the optimal control process.
Therefore, by part (i),
 $$
 \begin{array}{lll}
 &&\displaystyle |\widehat{H}(t_1,y,q)-\widehat{H}(t_2,y,q)|= \widehat{H}(t_1,y,q)-\widehat{H}(t_2,y,q)\\
 &=&\displaystyle  \mathbb E\left[\int_{t_1}^{t_2} L(\theta_r^*,X_{t_1,q}^{\theta^*}(r)) dr\big|\mathcal F_{t_1}\right]+\mathbb E\big[  \widehat{H}(t_2,Y_{t_1, y}(t_2), X_{t_1,q}^{\theta^*}(t_2))-\widehat{H}(t_2,y,q)\big|\mathcal F_{t_1}\big]\\
 &=: &I_1+I_2.
 \end{array}
 $$
%Noticing the fact that
%\begin{equation}\label{d12}
%X_{t_1,q}^{\theta^*}(r)=q-\int_{t_1}^{r}\theta^*_rdr<q.
%\end{equation}
% Since $\theta^*(\cdot)\in \Theta_{t_1}$,
%\begin{equation}\label{d22}
%\mathbb E\left[\max_{t_1\le r\le t_2}|\theta_r^*|\right]<\infty.
%\end{equation}
%These
%Using the argument  in Step 1, one can prove    that there exists a quadratic-growth    coefficient $C_1(q)$ such that
% \begin{equation}\label{I1}
% I_1\le C_1(q)|t_2-t_1|.
% \end{equation}
For the second term $I_2$, we have
 \begin{equation}\label{I2}
 \begin{array}{lll}
 I_2
 &\le&\displaystyle \mathbb E\left[\big| \widehat{H}(t_2,Y_{t_1, y}(t_2), X_{t_1,q}^{\theta^*}(t_2))-\widehat{H}(t_2,Y_{t_1, y}(t_2),q)\big|\Big|\mathcal F_{t_1}\right]\\
 &&\displaystyle +\mathbb E\left[\big| \widehat{H}(t_2,Y_{t_1, y}(t_2), q)-\widehat{H}(t_2,y,q)\big|\Big|\mathcal F_{t_1}\right]\\
& \le &\mathbb E[K_2(X_{t_1,q}^{\theta^*}(t_2), q)|X_{t_1,q}^{\theta^*}(t_2)-q||\mathcal F_{t_1}]+C_2(q) \mathbb E[|\ln(Y_{t_1,y}(t_2))-\ln(y)||\mathcal F_{t_1}]\\
& \le & C_1(q)\mathbb E[|X_{t_1,q}^{\theta^*}(t_2)-q||\mathcal F_{t_1}]+C_2(q) \mathbb E[|\ln(Y_{t_1,y}(t_2))-\ln(y)||\mathcal F_{t_1}],
 \end{array}
 \end{equation}
  where $C_1(q)$ and $C_2(q)$ are two  $(t,y)$-independent,    polynomial-growth coefficients.
The second-to-last inequality follows form the results in   Eq. (\ref{hidden}) and
Eq. (\ref{proofA2}).  The last inequality follows from the fact that
 $|X_{t_1,q}^{\theta^*}(t_2)|\le q$.
   Noticing that
  \begin{enumerate}
\item
by the completing square trick as used  in Step 1,
$$
\begin{array}{lll}
\displaystyle  I_1+C_1(q)\mathbb E[|X_{t_1,q}^{\theta^*}(t_2)-q||\mathcal F_{t_1}]
&=&\displaystyle  \mathbb E\left[ \int_{t_1}^{t_2} \left[L(\theta_r^*,X_{t_1,q}^{\theta^*}(r))+C_1(q)\theta^*_r\right]dr\Big|\mathcal F_{t_1}\right]\\
&\le &\displaystyle
 \Big( \frac{[(2\phi-\eta) +C_1(q)/q]^2}{4\nu}-\gamma \sigma^2\Big)q^2 |t_2-t_1|.
 %&=&\displaystyle \left\{ \Big(\frac{[2\phi-\eta]^2}{4\nu}-\gamma\sigma^2\Big)q^2+\frac{[2\phi-\eta]C_1(q) q}{2\nu}+\frac{C_1(q)^2}{4\nu}\right\}|t_2-t_1|.
 \end{array}
$$
 and that
\item $\displaystyle \mathbb E[|\ln(Y_{t_1, y}(t_2)-\ln(y)| |\mathcal F_{t_1}]=\mathbb E\left[|Z_{t_2-t_1}| \right]$, where
$$
Z_{t_2-t_1}=(\beta-\frac{\xi^2}{2})(t_2-t_1)+\xi W^Y_{t_2-t_1}
$$
is a normally distributed random variable with mean $(\beta-\frac{\xi^2}{2})(t_2-t_1)$ and variance $\xi^2(t_2-t_1)$.  Let $f_z(x)$ be the  probability density function of $Z_{t_2-t_1}$, then
$$
\begin{array}{lll}
\mathbb E\left[|Z_{t_2-t_1}|\right]&= &\displaystyle\int_{-\infty}^{\infty} |x|f_z(x)dx\\

  &\le &\displaystyle\sqrt{\int_{-\infty}^{\infty} x^2f_z(x)dx}\cdot \sqrt{\int_{-\infty}^{\infty} f_z(x) dx}  \\
  &=& \displaystyle\sqrt{ \mathbb E[Z_{t_2-t_1}^2]}\\
  &=&\displaystyle \sqrt{Var(Z_{t_2-t_1})+\mathbb E[Z_{t_2-t_1}]^2}\\
  &=&\displaystyle\sqrt{ \xi^2(t_2-t_1)+(\beta-\frac{\xi^2}{2})^2(t_2-t_1)^2}.
\end{array}
$$
Therefore, by Inequality (\ref{I2}),  there exists some polynomial-growth  coefficient $K_3(q)$ such that
$$
|\widehat{H}(t_1, y,q)-\widehat{H}(t_2, y,q)|=I_1+I_2\le K_3(q)(\sqrt{|t_2-t_1|}+|t_2-t_1|).
$$
\end{enumerate}
Combining the results in Steps 1, 2 and 3, we conclude that $\widehat{H}$ is locally H\"{o}lder continuous in $t$ with exponential 1/2, and locally Lipschitz continuous in both $y$ and $q$.
Since $H(t,y,q)=-\phi q^2+\widehat{H}(t,y,q)$, we conclude that
$H$ has the same continuity property in $[0, T)\times(\alpha^*, \infty)\times (0,+\infty)$.

\item[(iii)]
Since
$$
H(t,y,q)=-\phi q^2+\max_{\theta(\cdot)\in\Theta_t}\mathbb E\left[\int_t^{\tau_{t, m(y)}-}L( \theta_r,X_{t,q}^{\theta}(r))dr \Big|\mathcal F_t\right],
$$
by the completing square trick as used in part (ii), Step 1, we have
$$
\begin{array}{lll}
|H(t,y,q)|&\le&\displaystyle  \phi q^2+ \big[\frac{(2\phi-\eta)^2}{4\nu}-\gamma\sigma^2\big] q^2\mathbb E\left[ \tau_{t, m(y)}-t|\mathcal F_t\right]\\
&\le&\displaystyle  \left(\phi + \big[\frac{(2\phi-\eta)^2}{4\nu}-\gamma\sigma^2\big](T-t)\right) q^2.
\end{array}
$$
%According to the result in part (i) {\color{red} and the fact that $\lim_{y\to\infty} H(t, y, q)=U(t, q)$}, we have
%$$
%H(t,\alpha^*, q)\le  H(t,y,q)\le U(t, q), \quad \mbox{for any  $y> \alpha^*$}.
%$$
%Hence,
%$$
%0\le \widehat H(t, y, q)\le U(t, q)+\phi q^2=[c(t)+\phi]q^2
%$$
%where  $c(t)$ is  given in Eq. (\ref{r1}).
That is,  $H(t,y,q)$ satisfies a quadratic  growth condition with respect to  the inventory variable $q$, and is bounded in any compact subset of $[0,T)\times (\alpha^*, \infty)\times (0, \infty)$.

\end{description}

\subsection{Viscosity Solutions}

In Section 5.3, we discussed in detail the continuity of $H(t, y,q)$.
Since we do not expect  the value function $H$ to be  continuously differentiable,
we cannot discuss the solution to the HJB equation (\ref{HJB2}) in the classical sense.
Therefore we would like to introduce the concept of a viscosity solution.

\begin{definition}\label{definition2}
 A continuous function $ H(\cdot,\cdot,\cdot)$, on $[0,T)\times (\alpha^*,\infty)\times (0,+\infty)$ is a viscosity sub-solution (resp. super-solution) of the HJB equation (\ref{HJB2}), if  for any  $\mathcal C^{1,2,2}([0,T)\times (\alpha^*,\infty)\times (0,+\infty))$ function $\psi$ and  $( \bar{t},\bar{y},\bar{q})\in [0,T)\times (\alpha^*,\infty)\times (0,+\infty)$ such that $H( t,y,q)-\psi(t,y,q)$ attaints its local  maximum (resp. minimum) at  $( \bar{t},\bar{y},\bar{q})$, we have
$$
\begin{array}{lll}
\displaystyle -\Big(\partial_t+\beta y\partial_y+\frac{1}{2}\xi^2y^2\partial_{yy}\Big) \psi( \bar{t},\bar{y},\bar{q})+ \gamma  \sigma^2 \bar{q}^2\\\
\displaystyle\; -  \max_{\theta_t\in\Theta_t}\Big\{ g(\theta_t)\theta_t+ f(\theta_t) \bar{q}-\partial_{q}\psi( \bar{t},\bar{y},\bar{q})\cdot \theta_t\Big\}\le 0; \quad\mbox{(resp. $\ge 0$)}.
\end{array}
 $$
The continuous function $H$ is a viscosity solution if it is both a viscosity sub-solution and a viscosity super-solution.
\end{definition}

For the value function $H(t,y,q)$, we have the following results:
\begin{theorem}\label{theorem8}
The value function  $H$ is a viscosity solution of the HJB equation (\ref{HJB2}).
\end{theorem}

\begin{proof}
We will prove that $H$ is a viscosity super-solution and sub-solution of
Eq. (\ref{HJB2}) in Steps 1 and 2, respectively.

%  We first show that $H$ is a viscosity solution of the HJB equation (\ref{HJB2}). \\
 \begin{description}
\item[ Step 1:  $H$ is a viscosity super-solution of the HJB equation (\ref{HJB2}).]\quad\\
  Without loss of generality,
 let
 \begin{equation}\label{Viscosity}
 \min_{(t,y,q)\in[0,T)\times (\alpha^*,\infty)\times (0,+\infty)}(  H-\psi)(t,y,q)=(  H-\psi)(\bar{t},\bar{y},\bar{q})=0.
 \end{equation}
Assume that $\delta$ is sufficiently small   such that
$$
B_{\delta}(\bar{y},\bar{q}):=\{(y,q): \sqrt{(y-\bar y)^2+(q-\bar q)^2} <\delta\}
\subseteq (\alpha^*,+\infty)\times (0, +\infty).
$$
For any arbitrary  constant control $  \bar{\theta} \in \Theta_{\bar{t}}$, define
 $$
 \widehat{\tau}(\bar\theta)=\inf\{t\ge \bar{t}: (Y_{\bar{t}, \bar{y}}(t),X_{\bar{t}, \bar{q}}^{\bar \theta}(t))\notin B_\delta(\bar{y},\bar{q})\}.
 $$
 For any $0<\Delta t<T-\bar{t}$, define the stopping time
$$
\tau(\bar \theta, \Delta t)=(\bar{t}+\Delta t)\land\widehat{\tau}(\bar \theta)\land \tau_{\bar t, m(\bar{y})}.
$$
By DP principle,
 $$
 \begin{array}{lll}
  H(\bar{t},\bar{y},\bar{q})&\ge&\displaystyle   \mathbb E\left[\int_{\bar{t}}^{\tau(\bar \theta, \Delta t)-} \big[ g(\bar\theta)\bar\theta+f(\bar\theta)X_{\bar{t},\bar{q}}^{\bar\theta}(r)-\gamma \sigma^2 (X_{\bar{t},\bar{q}}^{\bar\theta}(r))^2\big]dr\right]\\
 &&\displaystyle  +\mathbb E\left[H(\tau(\bar\theta, \Delta t), Y_{\bar{t},\bar{y}}(\tau(\bar \theta, \Delta t)), X_{\bar{t},\bar{q}}^{\bar\theta}(\tau(\bar \theta, \Delta t)))\right].
 \end{array}
 $$
Eq. (\ref{Viscosity}) implies that $H(t,y,q)\ge \psi(t,y,q)$
and $H(\bar{t},\bar{y},\bar{q})=\psi(\bar{t},\bar{y},\bar{q})$, thus
 $$
 \begin{array}{lll}
 \psi(\bar{t},\bar{y},\bar{q})&\ge&\displaystyle   \mathbb E\left[\int_{\bar{t}}^{\tau(\bar\theta, \Delta t)-} \big[ g(\bar\theta)\bar\theta+f(\bar\theta)X_{\bar{t},\bar{q}}^{\bar\theta}(r)-\gamma \sigma^2 (X_{\bar{t},\bar{q}}^{\bar\theta}(r))^2\big]dr\right]\\
 &&\displaystyle  +\mathbb E\left[\psi(\tau(\bar \theta, \Delta t), Y_{\bar{t},\bar{y}}(\tau(\bar\theta, \Delta t)), X_{\bar{t},\bar{q}}^{\bar\theta}(\tau(\bar \theta, \Delta t)))\right].
 \end{array}
 $$
 Applying It\^{o}'s formula to $\psi(t, Y_{\bar{t},\bar{y}}(t), X_{\bar{t},\bar{q}}^{\bar\theta}(t))$ between $\bar{t}$ and $\tau(\bar\theta, \Delta t)$, we obtain
 \begin{equation}\label{inside}
 \begin{array}{lll}
\displaystyle   \mathbb E\Big[\frac{1}{\tau(\bar\theta, \Delta t)-\bar t}\int_{\bar{t}}^{\tau(\bar\theta, \Delta t)} \big[ g(\bar\theta)\bar\theta+f(\bar\theta)X_{\bar{t},\bar{q}}^{\bar\theta}(r)-\gamma \sigma^2 (X_{\bar{t},\bar{q}}^{\bar\theta}(r))^2\\
 \displaystyle\quad\quad\quad\quad \quad\quad+\mathcal L\psi(r, Y_{\bar{t},\bar{y}}(r), X_{\bar{t},\bar{q}}^{\bar\theta}(r))-\partial_q\psi(r, Y_{\bar{t},\bar{y}}(r), X_{\bar{t},\bar{q}}^{\bar\theta}(r))\cdot \bar\theta\;\big]dr\Big]\le 0,
  \end{array}
 \end{equation}
 where
 $$
 \mathcal L=\partial_t+\beta y\partial_y+\frac{1}{2}\xi^2y^2\partial_{yy}.
 $$
By the mean-value theorem, the random variable in the expectation (\ref{inside}) converges a.s. to
 $$
\mathcal L\psi(\bar{t}, \bar{y}, \bar{q})  -\gamma \sigma^2 \bar{q}^2+g(\bar\theta)\bar\theta+f(\bar\theta) \bar{q}  -\partial_q \psi(\bar{t},\bar{y},\bar{q})\cdot \bar{\theta}
 $$
as\footnote{For any arbitrary  constant control $  \bar{\theta} \in \Theta_{\bar{t}}$,
$\lim_{\Delta t\to0} \tau(\bar\theta, \Delta t)=\bar t$.} $\Delta t\to 0^+$. We then obtain
$$
  \mathcal L\psi(\bar{t}, \bar{y}, \bar{q})  -\gamma \sigma^2 \bar{q}^2+g(\bar\theta)\bar\theta+f(\bar\theta) \bar{q} -\partial_q \psi(\bar{t},\bar{y},\bar{q})\cdot \bar{\theta} \le 0.
$$
We conclude the proof from the arbitrariness of $\bar\theta\in \Theta_{\bar{t}}$. \\

\item[Step 2: $H$ is a viscosity sub-solution of the HJB equation (\ref{HJB2}).]\quad \\
Without loss of generality,  let
 \begin{equation}\label{Viscosity1}
 \max_{(t,y,q)\in[0,T)\times (\alpha^*,\infty)\times (0,+\infty)}(  H-\psi)(t,y,q)=(  H-\psi)(\bar{t},\bar{y},\bar{q})=0.
 \end{equation}
 We will show the result by contradiction.  Assume on the contrary that
$$
  \displaystyle \mathcal L \psi( \bar{t},\bar{y},\bar{q}) - \gamma  \sigma^2 \bar{q}^2 +  \max_{\theta_t}\Big\{ g(\theta_t)\theta_t+ f(\theta_t) \bar{q}-\partial_{q}\psi( \bar{t},\bar{y},\bar{q})\cdot \theta_t\Big\}< 0.
$$
Since $\psi\in \mathcal C^{1,2,2}([0,T)\times (\alpha^*,+\infty)\times (0,+\infty))$,
there exist $\delta>0$ and $\xi>0$ such that
 \begin{equation}\label{as10}
  \displaystyle \mathcal L \psi(t,y,q)- \gamma  \sigma^2 q^2 +  \max_{\theta_t}\Big\{ g(\theta_t)\theta_t+ f(\theta_t) q-\partial_{q}\psi( t,y,q)\cdot \theta_t\Big\}<- \xi,
 \end{equation}
for any   $(t,y,q)\in B_\delta(\bar t,\bar y,\bar q)$.  Here
$$
\begin{array}{lll}
\displaystyle B_\delta(\bar t,\bar y,\bar q)&:=&\displaystyle \{(t, y,q): \sqrt{(t-\bar t)^2+(y-\bar y)^2+(q-\bar q)^2} <\delta\}
\end{array}
$$
 is a 3-dimensional ball of radius $\delta$.
Without loss of generality, we can always choose $\delta$  to be sufficiently  small
so  that $B_\delta(\bar t,\bar y,\bar q) \subseteq [0,T)\times (\alpha^*,\infty)\times (0,+\infty)$.
For any arbitrary control process $\theta\in \Theta_{\bar t}$,  we define
 $$
 \widetilde{\tau}(\theta)=\inf\{t\ge \bar{t}: (t, Y_{\bar{t}, \bar{y}}(t),X_{\bar{t}, \bar{q}}^{\theta}(t))\notin B_\delta(\bar t, \bar{y},\bar{q})\}.
 $$
 For  any $0<\Delta t<T-\bar{t}$, define
$$
\tau^\prime(\theta, \Delta t) =(\bar{t}+\Delta t)\land\widetilde{\tau}(\theta)\land \tau_{\bar t,m(\bar{y})}.
$$
By the  DP principle, there exists a control process $\theta^\prime\in \Theta_{\bar t}$ such that
 $$
 \begin{array}{lll}
\displaystyle H(\bar t,\bar y, \bar q)-\frac{\xi}{2}\Delta t&\le &\displaystyle   \mathbb E\left[\int_{\bar{t}}^{\tau^\prime(\theta^\prime,\Delta t)-} \big[ g(\theta_r^\prime)\theta_r^\prime+f(\theta^\prime_r)X_{\bar{t},\bar{q}}^{\theta^\prime}(r)-\gamma \sigma^2 (X_{\bar{t},\bar{q}}^{\theta^\prime}(r))^2\big]dr\right]\\
 &&\displaystyle  +\mathbb E\left[H(\tau^\prime(\theta^\prime, \Delta t), Y_{\bar{t},\bar{y}}(\tau^\prime(\theta^\prime,\Delta t)), X_{\bar{t},\bar{q}}^{\theta^\prime}(\tau^\prime(\theta^\prime, \Delta t )))\right].
 \end{array}
 $$
Eq. (\ref{Viscosity1}) implies that $H(t,y,q)\le \psi(t,y,q)$
and $H(\bar{t},\bar{y},\bar{q})=\psi(\bar{t},\bar{y},\bar{q})$, thus
$$
 \begin{array}{lll}
 \displaystyle \psi(\bar t,\bar y, \bar q)-\frac{\xi}{2}\Delta t&\le &\displaystyle   \mathbb E\left[\int_{\bar{t}}^{\tau^\prime(\theta^\prime,\Delta t)-} \big[ g(\theta_r^\prime)\theta_r^\prime+f(\theta^\prime_r)X_{\bar{t},\bar{q}}^{\theta^\prime}(r)-\gamma \sigma^2 (X_{\bar{t},\bar{q}}^{\theta^\prime}(r))^2\big]dr\right]\\
 &&\displaystyle  +\mathbb E\left[\psi(\tau^\prime(\theta^\prime,\Delta t), Y_{\bar{t},\bar{y}}(\tau^\prime(\theta^\prime,\Delta t)), X_{\bar{t},\bar{q}}^{\theta^\prime}(\tau^\prime(\theta^\prime,\Delta t)))\right].
 \end{array}
 $$
Applying It\^{o}'s formula to $\psi(t, Y_{\bar{t},\bar{y}}(t), X_{\bar{t},\bar{q}}^{\theta^\prime}(t))$ between $\bar{t}$ and $\tau^\prime(\theta^\prime,\Delta t)$, we obtain
 \begin{equation}\label{inside1}
 \begin{array}{lll}
\displaystyle-\frac{\xi}{2}\le    \mathbb E\Big[ \frac{1}{\Delta t}\int_{\bar{t}}^{\tau^\prime(\theta^\prime, \Delta t)}\big[ g(\theta_r^\prime)\theta_r^\prime+f(\theta^\prime_r)X_{\bar{t},\bar{q}}^{\theta^\prime}(r)-\gamma \sigma^2 (X_{\bar{t},\bar{q}}^{\theta^\prime}(r))^2\\
\quad \quad\quad\quad\quad\quad\quad \displaystyle +\mathcal L\psi(r, Y_{\bar{t},\bar{y}}(r), X_{\bar{t},\bar{q}}^{\theta^\prime}(r))-\partial_q\psi(r, Y_{\bar{t},\bar{y}}(r), X_{\bar{t},\bar{q}}^{\theta^\prime}(r))\cdot \theta^\prime_r\big]dr\Big].
  \end{array}
 \end{equation}
 Letting $\Delta t\to0$,  Eq. (\ref{as10}) and Eq. (\ref{inside1})  imply that
  \begin{equation}\label{ed}
  -\frac{\xi}{2}\le -\lim_{\Delta t\to0}\frac{\mathbb E[\tau^\prime(\theta^\prime, \Delta t)-\bar t\; ]}{\Delta t}\xi.
  \end{equation}
  By the definition, we have
  $$
  \begin{array}{lll}
  \displaystyle 1\ge \frac{\mathbb E[\tau^\prime(\theta^\prime, \Delta t)-\bar t\; ]}{\Delta t}&\ge&\displaystyle  \frac{ \Delta t\times \mathbb P(\widetilde{\tau}(\theta^\prime)\land \tau_{\bar t, m(\bar{y})}>\bar t+\Delta t)}{\Delta t}\\
  &=&\displaystyle \mathbb P(\widetilde{\tau}(\theta^\prime) \land \tau_{\bar t, m(\bar{y})}>\bar t+\Delta t) \to1,
  \end{array}
  $$
  as $\Delta t\to 0^+$, which implies that
  $$
 \lim_{\Delta t\to 0} \frac{\mathbb E[\tau^\prime(\theta^\prime, \Delta t)-\bar t]}{\Delta t}=1.
 $$
 Eq. (\ref{ed}) then yields $-\frac{\xi}{2}\le -\xi$, which is  a contradiction. Therefore
 $$
 \displaystyle \mathcal L \psi( \bar{t},\bar{y},\bar{q}) - \gamma  \sigma^2 \bar{q}^2 +  \max_{\theta_t\in\Theta_t}\Big\{ g(\theta_t)\theta_t+ f(\theta_t) \bar{q}-\partial_{q}\psi( \bar{t},\bar{y},\bar{q})\cdot \theta_t\Big\}\ge0.
 $$
 \end{description}
Since the value function $H$ is both a viscosity sub-solution and a viscosity super-solution, we conclude that it is a viscosity solution of the HJB equation (\ref{HJB2}).

\end{proof}

\subsection{Comparison Principle and  Uniqueness}

The dynamic programming (DP) method is a powerful tool to study stochastic control problems by means of the HJB equation. However, in the classical approach, the method is used when it is assumed a priori that the value function is sufficiently smooth.
This, however, is not necessarily true even in some very simple cases.

To circumvent this difficulty, we adopt the  viscosity solutions approach in Section 5.4. In this section, we combine the results obtained in the previous sections with comparison principles for viscosity solutions. We characterize the value function as the unique viscosity solution of the associated dynamic programming equation, Eq. (\ref{HJB2}), and this can then be used to obtain further results.

\begin{theorem}\label{theorem9}
{\bf(Comparison Principle).}
Let $J^{sub}(t, y,q)$ (resp.  $J^{sup}(t,y,q)$)
% $\in\mathcal C^{1,2,1}([0,T)\times (\alpha^*,+\infty)\times (0,+\infty))\cap \mathcal C^0([0,T]\times(\alpha^*, +\infty)\times(0,+\infty))$\footnote{$f(t,y, q)\in\mathcal C^0([0,T]\times(\alpha^*, +\infty)\times(0,+\infty))$ means that $f(t, y, q)$ is continuous with $t$ over $[0,T]$ for any $(y,q)\in(\alpha^*, +\infty)\times(0,+\infty)$.}
 be  an  upper semi-continuous viscosity sub-solution  (resp. lower semi-continuous viscosity super-solution)  to Eq. (\ref{HJB2}), satisfying a polynomial growth condition with respect to  $q$.
Suppose  there exists a positive constant $r<1$ such  that   the growth rate of the  solution with respect to   $y$  can be controlled by  $\left[\ln\left(y\right)\right]^{r}$.
If
$$
\displaystyle J^{sub} \le J^{sup}  \ \ {\rm on} \ \
\{t=T\}\cup\{y=\alpha^*\}\cup\{q=0\},
$$
then
$$
J^{sub}\le J^{sup} \  \  {\rm in} \ \  [0,T)\times (\alpha^*, +\infty)\times (0,+\infty).
$$
\end{theorem}

\begin{proof} We complete the proof in the following steps:
\begin{description}
\item[Step 1. ]
Let  $\varrho>0$. Define $\widetilde J^{sub}=e^{\varrho t}J^{sub}$ and $\widetilde J^{sup}=e^{\varrho t}J^{sup}$.
A straightforward calculation shows that $\widetilde J^{sub}$ (resp. $\widetilde J^{sup}$) is a
sub-solution (resp. super-solution) to
\begin{equation}\label{bns}
\begin{array}{lll}
\displaystyle -\Big(\partial_t+\beta y\partial_y+\frac{1}{2}\xi^2y^2\partial_{yy}\Big) J( t, y,q )+e^{\varrho t} \gamma  \sigma^2q^2+\varrho J(t, y,q)\\
\displaystyle\; -  \max_{\theta_t\in\Theta_t}\Big\{ e^{\varrho t}[g(\theta_t)\theta_t+ f(\theta_t)q]-\partial_qJ( t,y,q)\cdot \theta_t\Big\}= 0.
\end{array}
 \end{equation}
 Define
$$
 \mathcal H(t, y,q,p, N)=\displaystyle \displaystyle \max_{\theta\ge0}\left\{b(y,q, \theta)\cdot p+\frac{1}{2} tr(\Sigma\Sigma^\prime(y) N)+f(t,y,q,\theta, p,N) \right\}
$$
where
  $$
  \left\{
\begin{array}{ll}
\displaystyle b(y,q,\theta) =\displaystyle \left(
\begin{array}{r}
\beta y\\
-\theta
\end{array}
\right)\\
\Sigma(y)=\displaystyle \left(
\begin{array}{cc}
\xi y& 0\\
0&0
\end{array}
\right) \\
f(t,y,q,\theta, p,N)=\displaystyle- e^{\varrho t}\gamma\sigma^2 q^2+e^{\varrho t}[g(\theta)\theta+f(\theta)q],
\end{array}
\right.
$$
$\Sigma^\prime$  is the transpose of $\Sigma$,  and $(p,\Sigma)=\big((p_1, p_2), \Sigma\big)\in\mathbb R^2\times \mathcal S_2$. Here
$\mathcal S_2$ is the set of symmetric $2\times 2$ matrices.
Eq. (\ref{bns})  can then  be rewritten as\footnote{We denote the gradient vector and matrix of second-order partial derivatives
of $J$, respectively,  by
$$
D_{(y,q)}J=\left(
\begin{array}{c}
J_y\\
 J_q
\end{array}
\right)\quad \mbox{and}\quad D^2_{(y,q)}J=\left(\begin{array}{cc}
J_{yy}&J_{yq}\\
J_{qy}&J_{qq}
\end{array}
\right).
$$}
 \begin{equation}\label{bns1}
-\partial_t J(t,y,q)+\varrho J(t,y,q)-\mathcal H(t,y,q, D_{(y,q)}J, D^2_{(y,q)}J)=0.
\end{equation}

\item[Step 2.] {\bf(Penalization and perturbation of super-solution)}
From the boundary and polynomial growth conditions on
$J^{sub}$ and $J^{sup}$, we may choose   an  integer $r_2>1$,
a positive constant $r_1<1$ and  $a>1/\alpha^*$   so that
 $$
 \sup_{(t,y, q)\in [0,T)\times(\alpha^*, +\infty)\times (0, +\infty)}\frac{|J^{sub}(t,y,q)|+|J^{sup}(t,y,q)|}{1+  \left[\ln\left(ay\right)\right]^{r_1}+ q^{r_2}}<\infty.
 $$
We then consider the function
$$
\Upsilon(t, y,q)=e^{-\varrho t}(1+\left[\ln(ay)\right]^{\frac{r_1+1}{2}} +q^{2r_2}).
$$
 A direct calculation shows that
 $$
 \begin{array}{lll}
&& \displaystyle (\partial_t+\beta y\partial_y+\frac{1}{2}\xi^2y^2\partial_{yy})  \Upsilon   \\
&=&\displaystyle -e^{-\varrho t}\Big\{\varrho-\frac{\beta(1+r_1)}{2[\ln(ay)]^{\frac{1-r_1}{2}}} +  \frac{\xi^2}{4} \left[\frac{1-r_1^2}{2}[\ln(ay)]^{\frac{r_1-3}{2}}+(1+r_1)[\ln(ay)]^{\frac{r_1-1}{2}}\right] \\
&&\displaystyle   + \varrho [\ln(ay)]^{\frac{r_1+1}{2}} +\varrho q^{2r} \Big\}\le 0,
\end{array}
 $$
as long as
$ \varrho>\max\left\{0,\frac{\beta(1+r_1)}{2[\ln(a\alpha^*)]^{\frac{1-r_1}{2}}}\right\}$.
This implies that for all $\omega>0$, the function $J^{sup}_\omega=J^{sup}+\omega \Upsilon$ is, as $J^{sup}$, a  super-solution to Eq. (\ref{HJB2}).  Furthermore, according to  the    growth conditions  on $J^{sub}$, $J^{sup}$ and   $\Upsilon$, we have, for all $\omega>0$,
 $$
 \lim_{\tiny
 \begin{array}{l}
 (y,q)\in (\alpha^*, +\infty)\times (0, +\infty)\\
  y+p\to +\infty
  \end{array}} \sup _{[0,T]} (J^{sub}-J^{sup}_\omega)(t,y,q)=-\infty.
 $$
 \item[Step 3.]
 Let $\mathcal D=(\alpha^*, +\infty)\times (0, +\infty)$. Without loss of generality, we may assume that the
 supremum of the upper semi-continuous function   $\widetilde J^{sub}-\widetilde J^{sup}$ on $\mathcal  [0,T]\times\bar{\mathcal D}$ is attained on $[0,T]\times (\mathcal O\cup \{y=\alpha^*\}\cup\{q=0\})$ for some   bounded open  set $\mathcal O$ of $(\alpha^*, +\infty)\times (0, +\infty)$. Here $\bar{\mathcal D}$ denotes the closure of  $\mathcal D$.   Otherwise, we can consider $\widetilde J^{sub}-\widetilde J_\omega^{sup}$  (that we described in Step 2)  instead of the original  $\widetilde J^{sub}-\widetilde J^{sup}$, and then passing to the limit $\omega\to 0^+$ at the end of the argument below to obtain the same result.

 Given  that $\widetilde J^{sub} \le \widetilde J^{sup} $ on $\{t=T\}\cup\{y=\alpha^*\}\cup\{q=0\}$, we need to prove that  $\widetilde J^{sub}\le \widetilde J^{sup}$ on $[0,T]\times [\alpha^*, \infty)\times[0,+\infty)$.  A rough framework for the  proof is  through  proof by contradiction.  We assume\footnote{ It is worth noting that under the hypothesis $\widetilde J^{sub} \le \widetilde J^{sup}$ on $\{t=T\}\cup\{y=\alpha^*\}\cup\{q=0\}$, the positive  maximum of $\widetilde J^{sub}-\widetilde{J}^{sup}$ cannot be attained on the boundary   $\{t=T\}\cup\{y=\alpha^*\}\cup\{q=0\}$.}
\begin{equation}\label{def}
M:=\sup_{[0,T]\times \bar{\mathcal D}}(\widetilde J^{sub}-\widetilde J^{sup})=\max_{[0,T)\times \mathcal O}(\widetilde J^{sub}-\widetilde J^{sup})>0.
\end{equation}

We now use the doubling of variables technique by considering,  for any $\epsilon >0$,
the function (for more details, please refer to Pham \cite{Pham09} for instance)
$$
 \Psi_\epsilon(t_1, t_2, y_1,y_2, q_1,q_2)=\displaystyle \widetilde J^{sub}(t_1, y_1,q_1)-\widetilde J^{sup}(t_2,y_2,q_2)-\psi_\epsilon(t_1, t_2, y_1,y_2, q_1,q_2)
$$
where
$$
\psi_\epsilon(t_1, t_2, y_1,y_2, q_1,q_2)=\displaystyle \frac{1}{\epsilon^2}(t_1-t_2)^2+\frac{1}{\epsilon}\left[(y_1-y_2)^2+(q_1-q_2)^2\right].
$$
The upper semi-continuous function $\Psi_\epsilon$ attains its maximum, denoted by $M_\epsilon$, on the compact set $[0,T]^2\times \bar{\mathcal O}^2$ at $(t^\epsilon_1, t^\epsilon_2, (y^\epsilon_1, q^\epsilon_1), (y^\epsilon_2, q^\epsilon_2))$.
Notice that
\begin{equation}\label{epsilon}
\begin{array}{lll}
M\le M_\epsilon&=&\displaystyle \Psi_\epsilon(t^\epsilon_1, t^\epsilon_2, y^\epsilon_1,y^\epsilon_2, q^\epsilon_1,q^\epsilon_2)\\
&=& \widetilde J^{sub}(t^\epsilon_1, y^\epsilon_1, q^\epsilon_1)-\widetilde J^{sup}( t^\epsilon_2, y^\epsilon_2, q^\epsilon_2)-\psi_\epsilon(t^\epsilon_1, t^\epsilon_2, y^\epsilon_1,y^\epsilon_2, q^\epsilon_1,q^\epsilon_2)\\
&\le& \widetilde J^{sub}(t^\epsilon_1, y^\epsilon_1, q^\epsilon_1)-\widetilde J^{sup}( t^\epsilon_2, y^\epsilon_2, q^\epsilon_2),
\end{array}
\end{equation}
for any $\epsilon >0$. It follows from the Bolzano-Weierstrass theorem  that there exists a subsequence of
 $(t^\epsilon_1, t^\epsilon_2, (y^\epsilon_1,q^\epsilon_1), (y^\epsilon_2,q^\epsilon_2))_\epsilon$  converging  to some point
$(\bar t_1, \bar t_2, (\bar y_1, \bar q_1),  (\bar y_2, \bar q_2))$ $\in[0,T]^2\times \bar{\mathcal O}^2$.  From  now on, we will consider such a convergent subsequence when necessary.
Moreover, since
the sequence $( \widetilde J^{sub}(t^\epsilon_1, y^\epsilon_1, q^\epsilon_1)-\widetilde J^{sup}( t^\epsilon_2, y^\epsilon_2, q^\epsilon_2))_\epsilon$ is bounded,  we see from Eq. (\ref{epsilon}) that the sequence $(\psi_\epsilon(t^\epsilon_1, t^\epsilon_2, y^\epsilon_1,y^\epsilon_2, q^\epsilon_1,q^\epsilon_2))_\epsilon$ is also bounded, and hence,  $\bar t_1=\bar t_2=\bar t$,  $\bar y_1=\bar y_2=\bar y$,  $\bar q_1=\bar q_2=\bar q$.
Passing to the limit  $\epsilon\to 0^+$ in Eq. (\ref{epsilon}), we obtain
 \begin{enumerate}
\item[(i)]  $M\le (\widetilde J^{sub}-\widetilde J^{sup})(\bar t_1, \bar y_1, \bar q_1)\le M$,
and hence $M=(\widetilde J^{sub}-\widetilde J^{sup})(\bar t_1, \bar y_1, \bar q_1)$ with $(\bar t_1, \bar y_1, \bar q_1) \in [0,T)\times \mathcal O$ by Eq. (\ref{def});
\item[(ii)]  As $\epsilon\to 0^+$, $M_\epsilon\to M$ and $\psi_\epsilon(t^\epsilon_1, t^\epsilon_2, y^\epsilon_1,y^\epsilon_2, q^\epsilon_1,q^\epsilon_2)\to 0$.
\end{enumerate}

Hence, by definition of $(t^\epsilon_1, t^\epsilon_2, y^\epsilon_1,y^\epsilon_2, q^\epsilon_1,q^\epsilon_2)$,
 \begin{description}
\item[(a)] $(t_1^\epsilon, y_1^\epsilon, q^\epsilon_1)$ is a local maximum of
$$
(t, y,q)\to \widetilde J^{sub}(t, y,q)-\psi_\epsilon(t, t^\epsilon_2, y, y^\epsilon_2, q, q^\epsilon_2)
$$
on $[0,T)\times \mathcal D $; and
\item[(b)] $(t_2^\epsilon, y_2^\epsilon, q^\epsilon_2)$ is a local minimum  of
$$
(t, y,q)\to \widetilde J^{sup}(t, y,q)+\psi_\epsilon( t^\epsilon_1, t,  y^\epsilon_1, y,  q^\epsilon_1, q)
$$
on $[0,T)\times \mathcal D $.
\end{description}

We define the second-order {\it superjets} $\mathcal P^{+, (1,2)} \widetilde J^{sub}(t,y,q)$ of $\widetilde J^{sub}$ at point $(t,y, q)\in [0,T)\times \mathcal D$, and the  second-order {\it subjets} $\mathcal P^{-, (1,2)} \widetilde J^{sup}(t,y,q)$ of $\widetilde J^{sup}$  as follows:
{\scriptsize
 $$
 \begin{array}{lll}
 \quad \displaystyle \mathcal P^{+,(1,2)} \widetilde J^{sub}(t, y,q)\\
 =\Big\{ (p_1, p_2, N)\in \mathbb R\times \mathbb R^2\times \mathcal S_2: \\
 \;\; \displaystyle \limsup_{\tiny\begin{array}{lll}
 ( \delta_t,  \delta_y, \delta_q)\to0\\
  (t+\delta_t, y+\delta_y, q+\delta_q)\in [0,T)\times\mathcal D
 \end{array}}\frac{ \widetilde J^{sub}(t+ \delta_t, y+\delta_y, q+\delta_q)-\widetilde J^{sub}(t, y,q)-p_1 \delta_t-p_2\cdot (\delta_y, \delta_q)-\frac{1}{2}N(\delta_y,\delta_q)^\prime\cdot(\delta_y,\delta_q)}{|\delta_t|+|\delta_y|^2+|\delta_q|^2}\le0\Big\},
 \end{array}
 $$
}and
{\scriptsize
 $$
 \begin{array}{lll}
 \quad \displaystyle \mathcal P^{-,(1,2)} \widetilde J^{sup}(t, y,q)\\
 =\Big\{ (p_1, p_2, N)\in \mathbb R\times \mathbb R^2\times \mathcal S_2: \\
 \;\; \displaystyle \limsup_{\tiny\begin{array}{lll}
 ( \delta_t,  \delta_y, \delta_q)\to0\\
 (t+\delta_t, y+\delta_y, q+\delta_q)\in [0,T)\times \mathcal D
 \end{array}}\frac{ \widetilde J^{sup}(t+ \delta_t, y+\delta_y, q+\delta_q)-\widetilde J^{sup}(t, y,q)-p_1 \delta_t-p_2\cdot (\delta_y, \delta_q)-\frac{1}{2}N(\delta_y,\delta_q)^\prime\cdot(\delta_y,\delta_q)}{|\delta_t|+|\delta_y|^2+|\delta_q|^2}\ge0\Big\},
 \end{array}
 $$
}
where
$\mathcal S_2$ is the set of symmetric $2\times 2$ matrices. From the definitions, we obtain
\begin{equation}\label{uo}
\begin{array}{l}
\displaystyle \left(\partial_{t_1} \psi_\epsilon, D_{(y_1,q_1)}\psi_\epsilon, D_{(y_1,q_1)}^2 \psi_\epsilon\right)(t_1^\epsilon, t_2^\epsilon, y^\epsilon_1, y^\epsilon_2, q^\epsilon_1, q_2^\epsilon)  \in \mathcal P^{+, (1,2)} \widetilde J^{sub}(t_1^\epsilon, y^\epsilon_1, q^\epsilon_1)\\
\displaystyle \left(-\partial_{t_2} \psi_\epsilon, -D_{(y_2,q_2)}\psi_\epsilon, -D_{(y_2,q_2)}^2 \psi_\epsilon\right)(t_1^\epsilon, t_2^\epsilon, y^\epsilon_1, y^\epsilon_2, q^\epsilon_1, q_2^\epsilon)  \in \mathcal P^{-, (1,2)} \widetilde J^{sup}(t_2^\epsilon, y^\epsilon_2, q^\epsilon_2).
\end{array}
\end{equation}
It is because (similar analysis can be applied to $\widetilde J^{sup}(t,y,q)$),
$$
\begin{array}{lll}
\widetilde J^{sub}(t,y,q)&\le&\displaystyle  \widetilde J^{sub}(t_1^\epsilon, y_1^\epsilon, q_1^\epsilon)+\psi_\epsilon(t, t_2^\epsilon, y, y_2^\epsilon, q, q_2^\epsilon)-\psi_\epsilon(t_1^\epsilon, t_2^\epsilon, y_1^\epsilon, y_2^\epsilon, q_2^\epsilon, q_2^\epsilon)\\
&=& \displaystyle  \widetilde J^{sub}(t_1^\epsilon, y_1^\epsilon, q_1^\epsilon)+\partial_{t_1}\psi_\epsilon(t_1^\epsilon, t_2^\epsilon, y_1^\epsilon, y_2^\epsilon, q_2^\epsilon, q_2^\epsilon)(t-t_1^\epsilon)\\
&&\displaystyle +D_{(y_1, q_1)}\psi_\epsilon(t_1^\epsilon, t_2^\epsilon, y_1^\epsilon, y_2^\epsilon, q_2^\epsilon, q_2^\epsilon) \cdot (y-y_1^\epsilon, q-q_1^\epsilon)\\
&&\displaystyle +\frac{1}{2}D^2_{(y_1, q_1)} \psi_\epsilon (t_1^\epsilon, t_2^\epsilon, y_1^\epsilon, y_2^\epsilon, q_2^\epsilon, q_2^\epsilon)(y-y_1^\epsilon, q-q_1^\epsilon)\cdot (y-y_1^\epsilon, q-q_1^\epsilon)\\
&&\displaystyle +o(|t-t_1^\epsilon|+|y-y_1^\epsilon|^2+|q-q_1^\epsilon|^2).
\end{array}
$$
Actually,   Eq. (\ref{uo}) holds  true for any  test    function $\psi\in \mathcal C^{1,2,2}([0,T)\times(\alpha^*, +\infty)\times (0, +\infty))$  of $\widetilde J^{sub}$ and $\widetilde J^{sup}$,   and the converse property also holds true: for any $(p_1, p_2, N)\in \mathcal P^{+, (1,2)} \widetilde J^{sub}(t, y,q)$, there exists $\psi\in \mathcal C^{1,2,2}([0,T)\times (\alpha^*, +\infty)\times (0,+\infty))$ satisfying
$$
(p_1,p_2, N)= \left(\partial_{t} \psi, D_{(y,q)}\psi, D_{(y,q)}^2 \psi\right)(t,  y, q)  \in \mathcal P^{+, (1,2)} \widetilde J^{sub}(t, y, q).
$$
We refer to Lemma 4.1 in Chapter V of  \cite{Fso06} for more details.  An equivalent definition of viscosity solutions in terms of {\it superjets} and {\it subjets} are given in the following proposition.

\begin{proposition}\label{proposition4}
An upper semi-continuous (resp. lower semi-continuous)  function $\omega$ on $[0,T)\times \mathcal D$ is a viscosity sub-solution (resp.   super-solution) of Eq. (\ref{bns1})  on   $[0,T)\times \mathcal D$ if and only if for all $(t,y,q)\in [0,T)\times \mathcal D$ and  all $(p_1, p_2, N)\in \mathcal P^{+, (1,2)} \omega(t,y,q)$ (resp.  $\mathcal P^{-, (1,2)} \omega(t,y,q)$),
$$
-p_1+\varrho \omega(t,y,q)-\mathcal H(t,y,q, p_2, N)\le \mbox{(resp. $\ge$)} \quad 0.
$$
\end{proposition}

The key tool in the comparison proof for second-order equations in the theory of viscosity solutions is a lemma in analysis due to Ishii.
We state this lemma without proof, and refer the reader to
Lemma 4.4.6 in \cite[P. 80]{Pham09} and
Lemma 3.6 in \cite[P. 32]{Koike} for more details.

\begin{lemma}{\bf(Ishii's lemma)}
Let $U$ (resp. $V$) be a upper semi-continuous (resp.  lower semi-continuous) function on $[0,T)\times \mathbb R^n$, $\psi\in \mathcal C^{1,1,2,2}([0,T)^2\times \mathbb R^n\times \mathbb R^n)$, and $(\bar t, \bar s, \bar x, \bar y)\in [0,T)^2\times \mathbb R^n\times \mathbb R^n$ a local maximum of $U(t, x)-V(s, y)-\psi(t,s,x,y)$.
Then, for all $\eta>0$, there exist $N_1, N_2\in \mathcal S_n$ satisfying
$$
\begin{array}{lll}
\displaystyle \left(\partial_t \psi, D_x\psi, N_1\right)(\bar t, \bar s, \bar x, \bar y)\in \mathcal P^{+, (1,2)}U(\bar t , \bar x),\\
\displaystyle \left(-\partial_s \psi, -D_y\psi, N_2\right)(\bar t, \bar s, \bar x, \bar y)\in \mathcal P^{-, (1,2)}V(\bar s , \bar y),\\
\end{array}
$$
and
$$
\begin{pmatrix}
N_1&0\\
0&-N_2
\end{pmatrix}
\le D_{(x, y)}^2\phi(\bar t, \bar s, \bar x, \bar y)+\eta \left(D_{(x,y)}^2\phi(\bar t ,\bar s, \bar x, \bar y)\right)^2.
$$
Here, $\mathcal S_n$ is the set of symmetric $n\times n$ matrices.
\end{lemma}

We shall use Ishii's lemma with
$$
\psi_\epsilon(t_1, t_2, y_1, y_2, q_1, q_2)= \frac{1}{\epsilon^2}(t_1-t_2)^2+\frac{1}{\epsilon}\left[(y_1-y_2)^2+(q_1-q_2)^2\right].
$$
Then, direct differentiations yield
\begin{equation}\label{eq1}
\left\{
\begin{array}{ll}
\displaystyle \partial_{t_1}\psi_\epsilon(t_1, t_2, y_1, y_2, q_1, q_2)=-\partial_{t_2}\psi_\epsilon(t_1, t_2, y_1, y_2, q_1, q_2)=\frac{2}{\epsilon^2}(t_1-t_2)\\
\displaystyle  \partial_{y_1}\psi_\epsilon(t_1, t_2, y_1, y_2, q_1, q_2)=-\partial_{y_2}\psi_\epsilon(t_1, t_2, y_1, y_2, q_1, q_2)=\frac{2}{\epsilon}(y_1-y_2)\\
\displaystyle  \partial_{q_1}\psi_\epsilon(t_1, t_2, y_1, y_2, q_1, q_2)=-\partial_{q_2}\psi_\epsilon(t_1, t_2, y_1, y_2, q_1, q_2)=\frac{2}{\epsilon}(q_1-q_2),
\end{array}
\right.
\end{equation}
$$
D_{((y_1, q_1),(y_2,q_2))}^2\psi(t_1, t_2, y_1, y_2, q_1, q_2)=\frac{2}{\epsilon}
\begin{pmatrix}
  I_2&-I_2\\
  -I_2&I_2
\end{pmatrix},
$$
and
$$
\begin{array}{lll}
\displaystyle \big(D_{(y_1, q_1), (y_2, q_2)}^2\psi(t_1, t_2, y_1, y_2, q_1, q_2)\big)^2
= \frac{8}{\epsilon^2}\begin{pmatrix}
I_2&-I_2\\
-I_2&I_2
\end{pmatrix}.
\end{array}
$$
Furthermore, by choosing $\eta=\epsilon$ in Lemma 1, there exist matrices $N_1$ and $N_2 \in S_2$ such that
\begin{equation}\label{eq2}
\begin{pmatrix}
N_1&0\\
0&-N_2
\end{pmatrix}
\le\frac{10}{\epsilon}
\begin{pmatrix}
 I_2&-I_2\\
  -I_2&I_2
\end{pmatrix}.
\end{equation}
In view of Ishii's lemma and Eq. (\ref{eq1}),
$$
\begin{array}{l}
\displaystyle \left(\frac{2}{\epsilon^2}(t_1^\epsilon-t_2^\epsilon), \; \; \Big(\frac{2}{\epsilon}(y_1^\epsilon- y_2^\epsilon)\;, \frac{2}{\epsilon}(q_1^\epsilon- q_2^\epsilon)\Big),\;\; N_1 \right)\in \mathcal P^{+, (1,2)} \widetilde J^{sub}(t_1^\epsilon, y_1^\epsilon, q_1^\epsilon),\\
\displaystyle \left(\frac{2}{\epsilon^2}(t_1^\epsilon-t_2^\epsilon), \; \; \Big(\frac{2}{\epsilon}(y_1^\epsilon-y_2^\epsilon)\;, \frac{2}{\epsilon}(q_1^\epsilon- q_2^\epsilon)\Big),\;\; N_2 \right)\in \mathcal P^{-, (1,2)} \widetilde J^{sup}(t_2^\epsilon, y_2^\epsilon, q_2^\epsilon).
\end{array}
$$

From viscosity sub-solution and super-solution characterization (i.e., Proposition 4) of $\widetilde J^{sub}$ and $\widetilde J^{sup}$ in terms of {\it superjets} and {\it subjets}, we then have\footnote{We use  the facts that
$\displaystyle \psi_\epsilon(t_1^\epsilon, t^\epsilon_2, y_1^\epsilon, y^\epsilon_2, q_1^\epsilon, q^\epsilon_2)=\displaystyle \widetilde J^{sub}(t_1^\epsilon, y_1^\epsilon, q_1^\epsilon)$
and
$ -\psi_\epsilon(t_1^\epsilon, t^\epsilon_2, y_1^\epsilon, y^\epsilon_2, q_1^\epsilon, q^\epsilon_2)=\displaystyle \widetilde J^{sup}(t_2^\epsilon, y_2^\epsilon, q_2^\epsilon) $.}
$$
\begin{array}{l}
\displaystyle
-\frac{2}{\epsilon^2}(t_1^\epsilon-t_2^\epsilon)  +\varrho \widetilde J^{sub}(t_1^\epsilon, y_1^\epsilon, q_1^\epsilon)-\mathcal H(t_1^\epsilon, y_1^\epsilon, q_1^\epsilon,   \Big(\frac{2}{\epsilon}(y_1^\epsilon- y_2^\epsilon), \frac{2}{\epsilon}(q_1^\epsilon- q_2^\epsilon)\Big), N_1)\le0,\\
\displaystyle-\frac{2}{\epsilon^2}(t_1^\epsilon-t_2^\epsilon) +\varrho \widetilde J^{sup}(t_2^\epsilon, y_2^\epsilon, q_2^\epsilon)-\mathcal H(t_2^\epsilon, y_2^\epsilon, q_2^\epsilon,   \Big(\frac{2}{\epsilon}(y_1^\epsilon- y_2^\epsilon), \frac{2}{\epsilon}(q_1^\epsilon-q_2^\epsilon)\Big), N_2)\ge0.
\end{array}
$$
By subtracting the above two inequalities, we get
\begin{equation}\label{ghj}
 \begin{array}{lll}
\displaystyle \varrho [\widetilde J^{sub}(t_1^\epsilon, y_1^\epsilon, q_1^\epsilon)-\widetilde J^{sup}(t_2^\epsilon, y_2^\epsilon, q_2^\epsilon)]
&\le& \displaystyle  \mathcal H\left(t_1^\epsilon, y_1^\epsilon, q_1^\epsilon,   \Big(\frac{2}{\epsilon}(y_1^\epsilon-y_2^\epsilon), \frac{2}{\epsilon}(q_1^\epsilon- q_2^\epsilon)\Big), N_1\right)\\
&&\displaystyle - \mathcal H\left(t_2^\epsilon, y_2^\epsilon, q_2^\epsilon,   \Big(\frac{2}{\epsilon}(y_1^\epsilon- y_2^\epsilon), \frac{2}{\epsilon}(q_1^\epsilon-q_2^\epsilon)\Big), N_2\right).
 \end{array}
 \end{equation}
Let us recall that
$$
 \mathcal H(t, y,q,p, N)= \displaystyle \frac{1}{2} tr(\Sigma\Sigma^\prime(y) N)+\beta y p_1-e^{\varrho t}\gamma\sigma^2 q^2+ \widehat {\mathcal H}(t, q, p_2),
$$
where $g(\theta)=-\nu\theta$, $f(\theta)=-\eta\theta$ and
 $$
 \begin{array}{lll}
\displaystyle  \widehat{\mathcal H}(t, q, p_2)&=&\displaystyle \max_{\theta\ge0}\left\{ e^{\varrho t}[g(\theta)\theta+f(\theta)q]-p_2\theta\right\}\\
&=&\displaystyle \left\{
 \begin{array}{lll}
 \displaystyle 0,&& \mbox{ if  $p_2\ge- \eta q e^{\varrho t}$}  \\
 \displaystyle \frac{1}{4\nu}e^{\varrho t}\left[\eta q+e^{-\varrho t}  p_2\right]^2,
&& \mbox{otherwise}.
 \end{array}
 \right.
 \end{array}
 $$

 The consequence on $\widehat {\mathcal H}$ is the crucial inequality
  \begin{equation}\label{ok23}
\left|\widehat{ \mathcal H}(t,  q, p_2)-\widehat{\mathcal H}(t^\prime, q^\prime, p_2)\right| \le C((1+|p_2|^2)|t-t^\prime| +(1+|p_2|)|q-q^\prime|),
 \end{equation}
 for all $(t,t^\prime,q,q^\prime,p_2)\in[0,T]^2\times(0, +\infty)^2\times \mathbb R$, where    $C$  is a constant depending  on    $\eta$, $\nu$, $\rho$, $T$ and $\mathcal O$.  We refer interested readers to Appendix A   for  the proof of Inequality (\ref{ok23}).
 Therefore,
 \begin{equation}\label{cont}
 \begin{array}{lll}
&&\displaystyle  \varrho [\widetilde J^{sub}(t_1^\epsilon, y_1^\epsilon, q_1^\epsilon)-\widetilde J^{sup}(t_2^\epsilon, y_2^\epsilon, q_2^\epsilon)]\\
 &\le&\displaystyle \frac{1}{2} tr(\Sigma\Sigma^\prime(y_1^\epsilon)N_1-\Sigma\Sigma^\prime(y_2^\epsilon)N_2)+ \frac{2 }{\epsilon}\beta|y_1^\epsilon-y_2^\epsilon|^2
 -e^{\varrho t_1^\epsilon}\gamma \sigma^2q_1^\epsilon+e^{\varrho t_2^\epsilon}\gamma \sigma^2q_2^\epsilon  \\
 &&\displaystyle +C \left(\frac{4}{\epsilon^2}|t_1^\epsilon-t_2^\epsilon|\cdot |q_1^\epsilon-q_2^\epsilon|^2+ |t_1^\epsilon-t_2^\epsilon|+|q_1^\epsilon-q_2^\epsilon|+\frac{2}{\epsilon}|q_1^\epsilon-q_2^\epsilon|^2\right).
 \end{array}
 \end{equation}
Due  to the  fact that  $\psi_\epsilon\to0$, as $\epsilon \to 0^+$, we have
$$
 |t_1^\epsilon-t_2^\epsilon|=o(\epsilon),\quad  |y_1^\epsilon-y_2^\epsilon|=o(\epsilon^{1/2})\quad \mbox{and}\quad |q_1^\epsilon-q_2^\epsilon|=o(\epsilon^{1/2}).
 $$
Hence,  as $\epsilon\to 0^+$,
$$
 \frac{4}{\epsilon^2}|t_1^\epsilon-t_2^\epsilon|\cdot |q_1^\epsilon-q_2^\epsilon|^2+|t_1^\epsilon-t_2^\epsilon|+|q_1^\epsilon-q_2^\epsilon|+\frac{2}{\epsilon}|q_1^\epsilon-q_2^\epsilon|^2 \to 0.
$$
We now use  Eq. (\ref{eq2}) to obtain
 \begin{equation}\label{cont1}
 \begin{array}{lll}
&&\displaystyle  \frac{1}{2} tr(\Sigma\Sigma^\prime(y_1^\epsilon)N_1-\Sigma\Sigma^\prime(y_2^\epsilon)N_2)\\
 &\le & \displaystyle \frac{1}{2} tr\left(\left[
 \begin{array}{cc}
 \Sigma\Sigma^\prime(y_1^\epsilon)& \Sigma(y_1^\epsilon)\Sigma^\prime(y_2^\epsilon)\\
 \Sigma(y_2^\epsilon)\Sigma^\prime(y_1^\epsilon)&  \Sigma\Sigma^\prime(y_2^\epsilon)
 \end{array}
 \right]\left[
 \begin{array}{cc}
 N_1&0\\
 0&N_2
 \end{array}
 \right]
 \right)\\
 &\le& \displaystyle \frac{5}{\epsilon}   tr\left(\left[
 \begin{array}{cc}
 \Sigma\Sigma^\prime(y_1^\epsilon)& \Sigma(y_1^\epsilon)\Sigma^\prime(y_2^\epsilon)\\
 \Sigma(y_2^\epsilon)\Sigma^\prime(y_1^\epsilon)&  \Sigma\Sigma^\prime(y_2^\epsilon)
 \end{array}
 \right]\left[
 \begin{array}{cc}
 I_2&-I_2\\
 -I_2&I_2
 \end{array}
 \right]
 \right)\\
 &=& \displaystyle \frac{5}{\epsilon} tr \big(\Sigma\Sigma^\prime(y_1^\epsilon)-\Sigma(y_1^\epsilon)\Sigma^\prime(y_2^\epsilon)- \Sigma(y_2^\epsilon)\Sigma^\prime(y_1^\epsilon)+ \Sigma\Sigma^\prime(y_2^\epsilon\big)\\
 &=&\displaystyle \frac{5}{\epsilon} tr\big([ \Sigma(y_1^\epsilon)-\Sigma(y_2^\epsilon)] [ \Sigma^\prime(y_1^\epsilon)-\Sigma^\prime(y_2^\epsilon)]
)\\
&=&\displaystyle \frac{5}{\epsilon} ||\Sigma(y_1^\epsilon)-\Sigma(y_2^\epsilon)||^2\\
&=&\displaystyle \frac{5}{\epsilon}  \xi^2 |y_1^\epsilon-y_2^\epsilon|^2.
  \end{array}
\end{equation}
 Combining the results in Eq. (\ref{cont}) and (\ref{cont1}),
  we conclude via passing to the limit $\epsilon\to0^+$ that $\varrho M\le 0$, which  contradicts to Eq.  (\ref{def}).  As a result, the assumption Eq. (\ref{def}) is false, and hence, the comparison principle $J^{sub}\le J^{sup}$ on $[0,T)\times\mathcal D$ holds.

 \end{description}
 \end{proof}

In Theorem \ref{theorem8}, we prove  that  $H(t,y,q)$ is a viscosity solution of the HJB equation (\ref{HJB2}).
In the proof of Theorem \ref{theorem7}, we verify that
\begin{description}
\item[(g1)] $H(t,y,q)$ satisfies a polynomial growth condition with respect to the inventory variable $q$; and that
\item[(g2)]  $H(t,y,q)$   can be controlled by a $y$-independent term:
$$
\left(\phi + \big[\frac{(2\phi-\eta)^2}{4\nu}-\gamma\sigma^2\big](T-t)\right) q^2.
$$
\end{description}
  Combining these results  with  the  comparison principle, we  can prove that  the value function is  the unique viscosity solution of Eq. (\ref{HJB2}). We provide  this result in the following corollary:
\begin{corollary} {\bf(Uniqueness).}
The value function  $H$ is the unique  viscosity solution of the HJB equation (\ref{HJB2}), satisfying the boundary and terminal conditions Eq. (35.a), (35.b) and (35.c), as well as the growth conditions (g1) and (g2).
\end{corollary}
\begin{proof}
Suppose  $H_1$ and $H_2$ are two viscosity solutions of the HJB equation,
Eq. (\ref{HJB2}), then
$H_1(T, y,q)=H_2(T,y,q)=-\phi q^2$,  $H_1(t, \alpha^*, q)=H_2(t, \alpha^*, q)=-\phi q^2$ and $H_1(t, y, 0)=H_2(t, y,0)=0$.
 According to Theorem \ref{theorem9}, if  we view $H_1$ as the sub-solution and $H_2$ as the super-solution, then
$$
H_1(\cdot,\cdot,\cdot)\le H_2(\cdot,\cdot,\cdot)
$$
holds over $[0,T)\times \mathcal D$; on the other hand,  if we view $H_2$ as the sub-solution and $H_1$ as the super-solution, then
$$
H_2(\cdot,\cdot,\cdot)\le H_1(\cdot,\cdot,\cdot)
$$
holds over $[0,T)\times \mathcal D$.  Hence $H_1(\cdot,\cdot,\cdot)\equiv H_2(\cdot,\cdot,\cdot)$.  According to Theorem \ref{theorem8}, the value function  $H(t,y,q)$ is a viscosity solution of Eq. (\ref{HJB2}), so it is the unique one.
\end{proof}

\subsection{Numerical Scheme }
%Many financial problems, such as price valuation and decision making,   do not have closed-form analytical solutions. Their valuations rely on a variety of numerical methods.  
In this section, we mainly discuss the application of finite difference methods on optimal liquidation  problems.

\subsubsection{The Value Function}

Similar to Section 3.2,  we consider an    ansatz of $H(t,y,q)$ that is quadratic in the variable $q$:
 \begin{equation}\label{H}
 H =  h^H(t,y) q^2.
 \end{equation}
With Assumption (\ref{H}), our  optimal liquidation strategy for the associated  unconstrained problem in Eq. (\ref{HJB2})  can   be written in the following feedback form:
 \begin{equation}\label{uc}
 \theta_t^{\phi,*}=-\frac{1}{2\nu} [2h^H(t,y)+\eta]q.
 \end{equation}
 According to the results in  Theorem \ref{theorem6} and Theorem \ref{theorem7}, we have
$$
H(t,\alpha^*, q)\le  H(t,y,q)\le U(t, q), \quad \mbox{for any  $y> \alpha^*$}.
$$
Therefore,
$$
 \theta_t^{\phi,*}=-\frac{1}{2\nu} [2h^H(t,y)+\eta]q\ge -\frac{1}{2\nu} [2c(t)+\eta]q\ge0.
$$
Similar arguments can then be applied here to prove that,
$$
 \int_{[0, T)} \theta_t^{\phi,*}dt =Q -X_{T-}^{\phi, *} \le Q.
 $$
 That is, the unconstrained optimal liquidation strategy (\ref{uc}) is the optimal one for the original  constrained problem.
The  unknown function  $h^H(t,y)$  then solves the following  PDE:
$$
\begin{array}{l}
\mbox{(I)}\quad\quad
\left\{
\begin{array}{lll}
 \displaystyle  \Big(\partial_t+\beta y\partial_y+\frac{1}{2}\xi^2y^2\partial_{yy}\Big) h^H  - \gamma  \sigma^2
+\frac{1}{4\nu} [2h^H(t,y)+\eta]^2=0\\
 \displaystyle h^H(T, y)=-\phi,\quad\quad\quad y>\alpha^*\\
 \displaystyle h^H(t, \alpha^*)=-\phi,\quad\quad\;\;\; 0\le t\le T\\
  \displaystyle \lim_{y\to \infty} h^H(t,y)=c(t)
\end{array}
\right.
\end{array}
$$
where   $c(t)$ is  given in Eq. (\ref{r1}).

\subsubsection{Numerical Scheme for $h^H$}

Define variables $x=\log \frac{y}{\alpha^*}$ and $\tau=T-t$. Set $\widetilde{h}^H=2h^H+\eta$.
The PDE satisfied by the unknown function $\widetilde{h}^H$ (similar methods can be used to the function $f^H$)  can then be written as
$$
\begin{array}{l}
\mbox{(I')}\quad\quad
\left\{
\begin{array}{lll}
 \displaystyle \partial_\tau \widetilde{h}^H=(\beta- \frac{\xi^2}{2})\partial_x \widetilde{h}^H+\frac{\xi^2}{2}\partial_{xx}\widetilde{h}^H  - 2\gamma  \sigma^2
+\frac{1}{2\nu} (\widetilde{h}^H)^2\\
\displaystyle \widetilde{h}^H(0,x)=-2\phi+\eta \\
\displaystyle \widetilde{h}^H(\tau, 0)=-2\phi+\eta\\
 \displaystyle \lim_{x\to \infty} \widetilde{h}^H(t,x)=2c(t)+\eta.
\end{array}
\right.
\end{array}
$$
To approximate the solution of the  PDE,
we discretize variables $\tau $ and $x$ with step sizes $\Delta \tau$  and $\Delta x$, respectively.
The value of $\widetilde h^H$ at a grid point $\tau_i=i\Delta \tau$ ($i=0,1,\cdots, N$ with $\tau_N=T$) and $x_j=j\Delta x$ is denoted by $\widetilde{h}^{H,i}_{j}$.
To approximate the infinite boundary condition
$$
\lim_{x\to \infty} \widetilde{h}^H(t,x)=2c(t)+\eta,
$$
we choose a sufficiently  large  integer $M$, and set $\widetilde h^{H,i}_{M}=2c^i_{M}+\eta$.
In this section,  we use the explicit difference method to perform  these numerical simulations:
$$
\left\{
 \begin{array}{l}
 \displaystyle \partial_\tau \widetilde{h}^H=(\widetilde{h}^{H, i+1}_j-\widetilde{h}^{H, i}_{j})/\Delta\tau\\
 \displaystyle \partial_x \widetilde{h}^H=(\widetilde{h}^{H,i}_{ j+1}-\widetilde{h}^{H, i}_{j-1})/(2\Delta x)\\
 \displaystyle \partial_{xx} \widetilde{h}^H=(\widetilde{h}^{H,i}_{j+1}+\widetilde{h}^{H, i}_{j-1}-2\widetilde{h}^{H,i}_{j})/\Delta x^2.
 \end{array}
\right.
 $$
To numerically solve the nonlinear PDE, we use a single Picard iteration, i.e., approximating a nonlinear term like
$(\widetilde{h}^{H,i}_{j})^2$ by $ \widetilde{h}^{H,i}_{j}\widetilde{h}^{H, i+1}_j$.
We define
$$
r=\Delta \tau\xi^2/\Delta x^2, \ \
v=r/2-\Delta\tau\big(\beta-\frac{\xi^2}{2}\big)/(2\Delta x), \ \
u=r/2+ \Delta\tau\big(\beta-\frac{\xi^2}{2}\big)/(2\Delta x).
$$
After implementing the scheme, we end up with a problem of solving linear system of equations:
   \begin{equation}
  B_i  \widetilde{{\bold h}}^{H, i+1}=A\widetilde{\bold h}^{H, i}-2\Delta \tau \gamma\sigma^2 \bold e
   \end{equation}
   where $\bold e=(1,1,\cdots, 1)^T$ is a $(M+1)$-dimensional  vector,  $B_i$ and $A$  are   square matrices of size $(M+1)$ given by
      $$
   B_i=\begin{pmatrix}
    1-\frac{\Delta \tau\widetilde h^{H,i}_{0}}{2\nu}&&\\
       &\ddots&\\
   && 1-\frac{\Delta \tau\widetilde h^{H, i}_{M}}{2\nu}
   \end{pmatrix}
   \quad {\rm and} \quad
    A=\begin{pmatrix}
 \displaystyle 1-r&u&&\\
   v&\displaystyle 1-r&u&\\
    &\ddots&\ddots&\\
   &&v&\displaystyle 1-r
   \end{pmatrix}.
 $$
The above system is solved for every time step moving forward in time given an initial conditions $\widetilde {\bold h}^{H, 0}=(h^{H, 0}_0,\cdots, h_{ M}^{H,0})^T$ and $B_0$.
Since the associated problem  has defined boundary conditions, for the solution vector $\widetilde{ \bold h}^{H, i}$ of size $(M+1)$, we only need to solve for the middle $M-1$ entries (i.e., $\widetilde h^{H, i}_{0}$ and $\widetilde h^{H, i}_{ M}$ are given by the boundary conditions) at every time iteration.

\subsubsection{Numerical Experiments}

To analyze the effectiveness of this numerical method, we will provide  a comparison of the numerical solution to the closed-form solution of   Model 1.
\begin{table}[H]
\begin{center}
\begin{tabular}{|l|c|c|}
\hline
Length of Time Interval & $T$ & $1$\\
\hline
Time Steps &$N$ & $1000$\\
  \hline
\end{tabular}
\caption{Parameters used in the implementation of the numerical schemes.}
\end{center}
\end{table}
\begin{figure}
\begin{center}
\includegraphics[width=6.3in, height=3.5in]{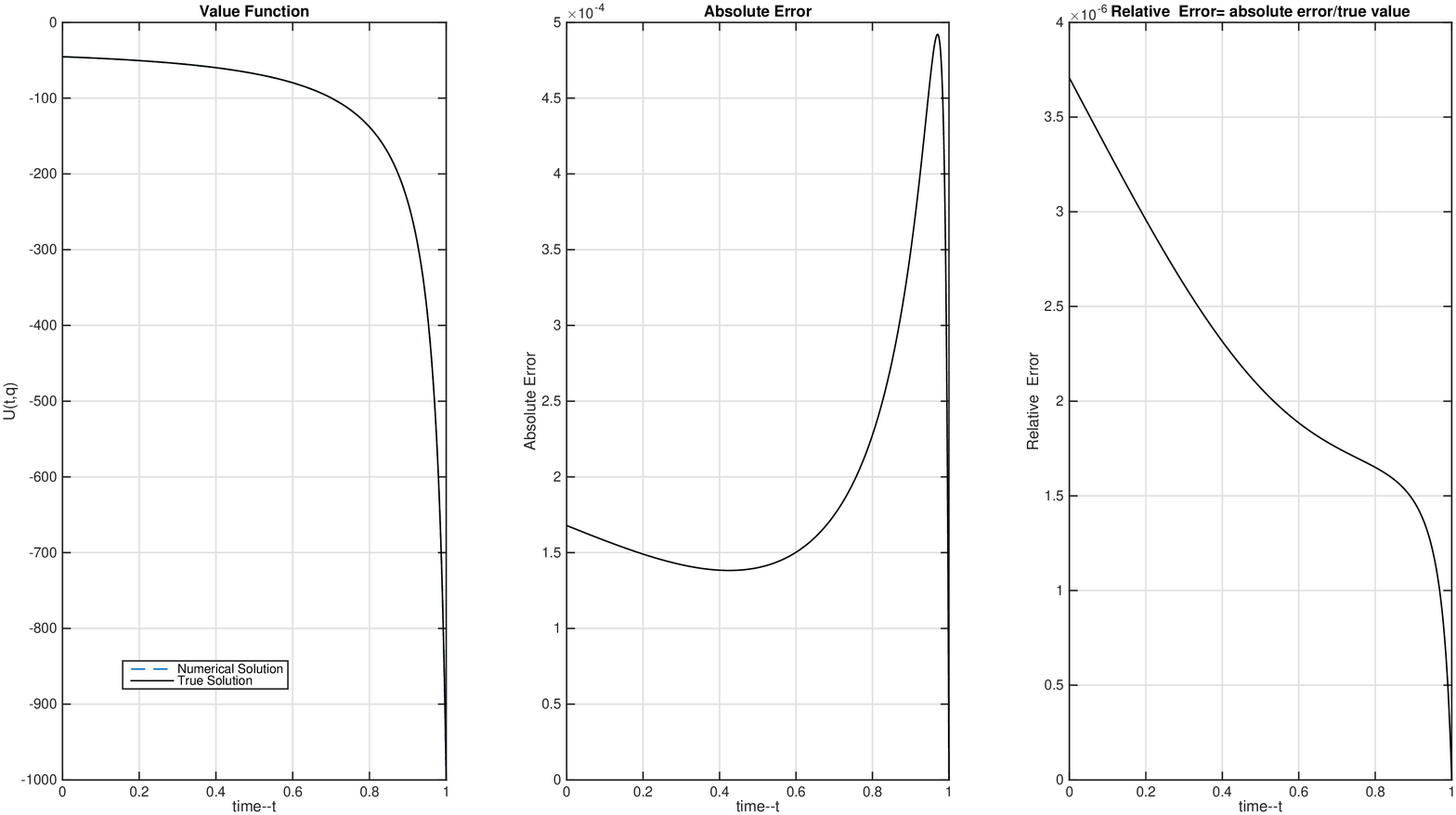}
\caption{ Comparison of the numerical and true solutions for Model 1 with $X_t\equiv Q=100$, for any $t\in[0,T]$. All parameters are the same with those used in Figure 1. }
 \label{error}
\end{center}
\end{figure}
 In the left plot\footnote{Corresponding  to the strategy  of inactive traders, such as buy and hold strategy. } given in Figure \ref{error}, we can observe that it is difficult to observe the discrepancies between the true solution of the HJB equation which is a decreasing function of time $t$,   and the numerical solution provided by our scheme. The plots of the  absolute  error and the corresponding relative error between the two solutions shows that, although there is some difference between the solutions, this  difference has a magnitude of $10^{-4}$ which is negligible given that  the actual solution is of magnitude $100$.  This motivates us to apply the numerical scheme to the optimal liquidation problem with default risk.

\subsubsection{Discretization of the Continuous Process}
Consider the stock issuer's market value $Y_t$.
Denote
$$
\wp_k=\left\{(k-\frac{1}{2})\Delta x<\log\left(\frac{Y_t}{\alpha^*}\right)<(k+\frac{1}{2})\Delta x\right\}.
$$
Set $\log\left(\frac{Y_t}{\alpha^*}\right)=k\Delta x$, while $\log\left(\frac{Y_t}{\alpha^*}\right)\in \wp_k$.
Suppose the initial market value is  $y=1000 \;m$,
where `$m$' represents the unit `million'.
We assume that the barrier $\alpha^*=10\; m$.
Table 2 displays the parameters used in the implementation of the numerical schemes.
Other model parameters not listed in Table 2
are the same with those used in Figure  1.

\begin{table}[H]
\begin{center}
\begin{tabular}{|l|c|c|}
\hline
Length of Time Interval &$T$ & $1$\\
\hline
Time Steps & $N$ & $1000$\\
\hline
Space Steps & $M$ &$10000$\\
\hline
Drift of Stock issuer's Market Value & $\beta$ & $-0.5$\\
\hline
Volatility of Stock issuer's Market Value & $\xi$ & $2$\\
\hline
\end{tabular}
\caption{Parameters used in the implementation of the numerical schemes.}
\end{center}
\label{1}
\end{table}

\begin{figure}
\begin{center}
\includegraphics[width=6.3in, height=3.5in]{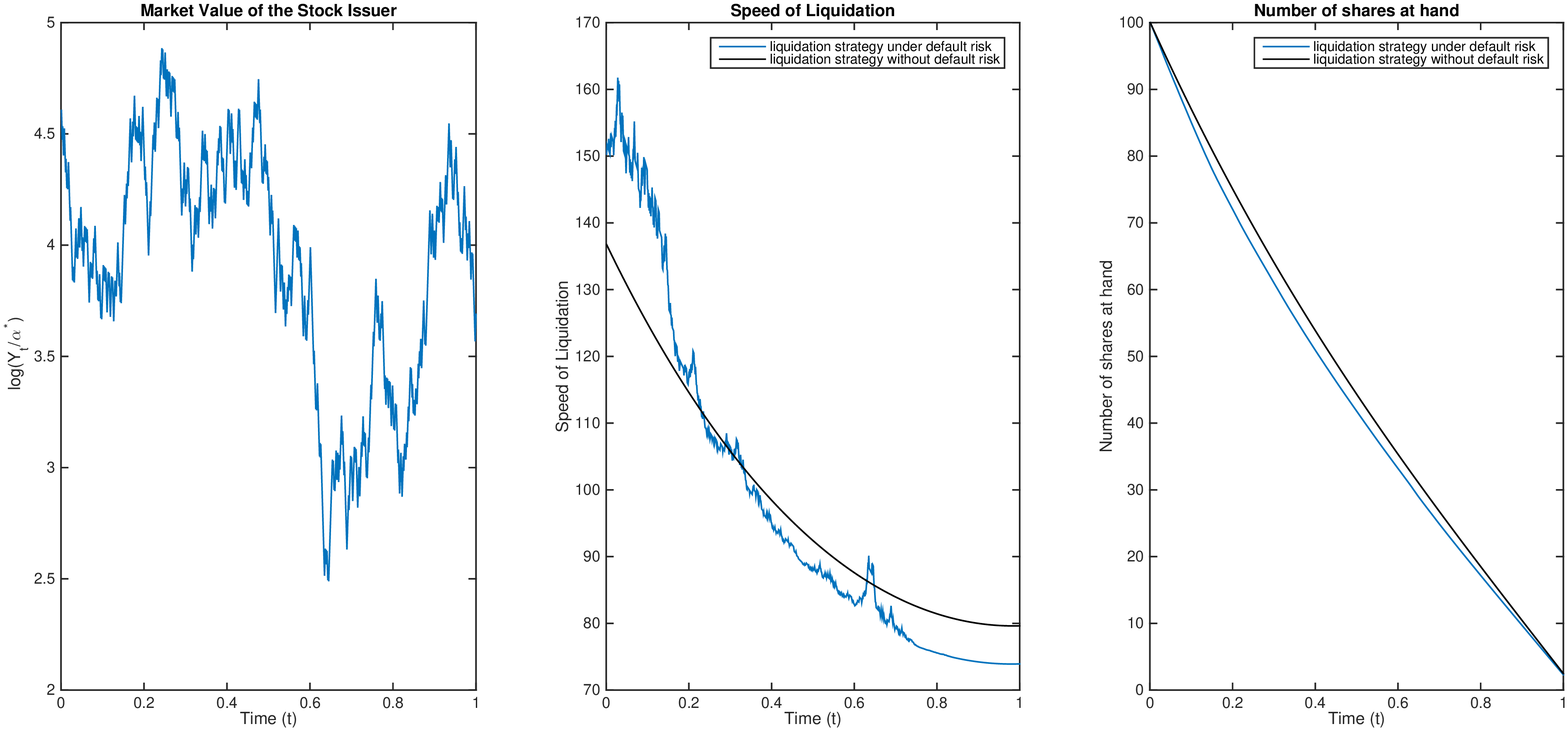}
\caption{A comparison of liquidation strategies with/without counter-party default risk. }
 \label{risk}
\end{center}
\end{figure}
 Figure \ref{risk} shows that the optimal liquidation strategy and the corresponding optimal number of shares in the risky asset  for one simulation of the  market value  path,  $\{\log\big(\frac{Y_t}{\alpha^*}\big),t\in[0,T]\}$.
We observe that strategies under default risk are   inter-temporarily updated. This update process, different from the first model,  depends  not only      on the  remaining time   to liquidate, but also  on the stock issuer's  market value.   Notice that, at time $t=0.6$, there is a sharp drop in   the stock issuer's market value:
 $$
\Delta  \log\left(\frac{Y_{0.6}}{\alpha^*}\right)=\log\left(\frac{Y_{0.6+}}{\alpha^*}\right)-\log\left(\frac{Y_{0.6}}{\alpha^*}\right) =2.5-4=-1.5,
$$
about 424  millions\footnote{
$Y_{0.6}\times(1- e^{-1.5})=\alpha^* e^4\times(1- e^{-1.5})$.}.
To reduce the exposure to the potential position risk incurred by the counter-party's default risk, the agent would therefore speed up his/her liquidation speed. This explains the local peak near $t=0.6$ in the middle plot of Figure \ref{risk}.

\section{Conclusions}

In this paper, we adopt Almgren-Chriss's market impact model and
 relax the assumption of a known pre-determined time horizon to study the  optimal liquidation problem  under a randomly-terminated setting. In some situation, it is more realistic to assume that the liquidation horizon  depends on some of the stochastic factors of the model.
For example, some financial markets adopt the circuit-breaking mechanism, which makes
the horizon of the investor subject to  stock price movement.
Once the stock price touches the daily limits, all transactions of the stock will be suspended.

Optimal liquidation strategy of large trades in a known pre-determined  time horizon is first discussed as a benchmark case.
We then extend our basic model to a randomly-terminated time horizon.
In particular, two different liquidation scenarios are analyzed to shed light on the relation between optimal liquidation strategies and potential liquidity risk subject to either
\begin{enumerate}
\item  an exogenous trigger event controlled by the hazard rate $\{l(t),t\ge0\}$; or
\item a  counterparty risk.
%\item  limits on price movements.
\end{enumerate}
For cases where no closed-form solutions can be obtained,
we study the problem via  the stochastic control approach.
By combining our results with comparison principles for viscosity solutions,
we characterize the value function as the unique viscosity solution of the associated HJB equation, and hence,
  the  optimal liquidation strategies that we found numerically serve as good approximations of the unique solutions according to the theory of viscosity solutions.

\section*{Acknowledgements}

This research work was supported by Research Grants Council of
Hong Kong under Grant Number 17301214 and  HKU Strategic Research Theme in Information and Computing  and National Natural Science Foundation of China under Grant number 11671158.
T. K. Wong is partially supported by the HKU Seed Fund for Basic Research under the project  code 201702159009, and the Start-up Allowance for Croucher Award Recipients.

\section*{Appendix. A}

In this appendix we will prove Inequality (\ref{ok23}), which is a technical and crucial inequality that we used in the proof of Theorem 9.\\

\noindent
{\bf Proof of Inequality (\ref{ok23}).}\\
There are three cases that we need to consider:
\begin{description}
\item[Case 1.] Both $\widehat{\mathcal H}(t, q, p_2)$ and   $\widehat{\mathcal H}(t^\prime, q^\prime, p_2)$ are equal to zero;
\item[Case 2.] Neither  $\widehat{\mathcal H}(t, q, p_2)$ nor    $\widehat{\mathcal H}(t^\prime, q^\prime, p_2)$ are   equal to zero;
\item[Case 3.] Only one of $\widehat{\mathcal H}(t, q, p_2)$ and   $\widehat{\mathcal H}(t^\prime, q^\prime, p_2)$  equals to zero.
\end{description}
It is obvious that  Inequality (\ref{ok23}) holds  under {\bf Case 1}, because the left hand side of (\ref{ok23}) equals to 0 in this case.  Therefore, we will only  focus on the proofs
for Cases 2 and 3.
\begin{description}
\item[Case 2.]  Without loss of generality, we assume that
$\widehat{\mathcal H}(t, q, p_2)\ge \widehat{\mathcal H}(t^\prime, q^\prime, p_2)>0$.
Hence,
\begin{equation}\label{lp}
\begin{array}{lll}
\displaystyle \widehat{\mathcal H}(t, q, p_2)- \widehat{\mathcal H}(t^\prime, q^\prime, p_2)
=  \frac{1}{4\nu}e^{\varrho t}[\eta q +e^{-\varrho t}p_2]^2- \frac{1}{4\nu}e^{\varrho t^\prime}[\eta q^\prime +e^{-\varrho t^\prime}p_2]^2\\
\quad\quad\quad\quad \quad\quad\quad\quad\quad\quad\;\; =\displaystyle \frac{1}{4\nu}\left[ \eta^2 [q^2 e^{\varrho t}-(q^\prime)^2e^{\varrho t^\prime}]
+2\eta p_2(q-q^\prime)+p_2^2\cdot [e^{-\varrho t}-e^{-\varrho t^\prime}]
\right].
\end{array}
\end{equation}
Notice that
$$
 q^2e^{\varrho t}-(q^\prime)^2 e^{\varrho t^\prime}=  q^2[e^{\varrho t}-e^{\varrho t^\prime}]+  e^{\varrho t^\prime}(q-q^\prime)(q+q^\prime).
$$
It follows  from the mean value theorem that there exist $\tilde{t}$ and $\hat{t}$ in between $t$ and $t^\prime$ such that
\begin{equation}\label{pol}
\left\{
\begin{array}{lll}
|e^{\varrho t}-e^{\varrho t^\prime}|=\displaystyle \left[\frac{d}{d t}e^{\varrho t}\right]\Big|_{t=\tilde t} |t-t^\prime|\le \varrho e^{\varrho T} |t-t^\prime|\\
|e^{-\varrho t}-e^{-\varrho t^\prime}|=\displaystyle\left| \left[\frac{d}{d t}e^{-\varrho t}\right]\Big|_{t=\hat t}\right|  |t-t^\prime|\le \varrho  |t-t^\prime|.
\end{array}
\right.
\end{equation}

Since $(t,t^\prime, (y, q), (y^\prime, q^\prime))$ belongs to the bounded set $[0, T)^2\times \bar{\mathcal O}^2$,
there exists a constant $\hat Q$ so that $|q|, |q^\prime|\le \hat Q$.
Therefore,
\begin{equation}\label{pl0}
|q^2e^{\varrho t}-(q^\prime)^2 e^{\varrho t^\prime}| \le  \varrho e^{\varrho T}  \hat Q^2 |t-t^\prime| +  2 \hat Qe^{\varrho T}|q-q^\prime|.
\end{equation}
Combining the results in Eqs. (\ref{lp}), (\ref{pol}),  and (\ref{pl0}),  we conclude that
$$
\widehat{\mathcal H}(t,q, p_2)-\widehat{\mathcal H}(t^\prime,q^\prime, p_2)\le C((1+ |p_2|^2)|t-t^\prime|+(1+|p_2|)|q-q^\prime|),
$$
where $C$ is a constant depending on $\eta$, $\nu$, $\rho$, $T$ and $\hat Q$.

\item[Case 3.] Without loss of generality, we assume that
$$
\widehat{\mathcal H}(t^\prime, q^\prime, p_2)=0\quad {\rm and }\quad \widehat{\mathcal H}(t, q, p_2)>0.
$$
That is,
\begin{equation}\label{fg}
-\eta q^\prime e^{\varrho t^\prime}\le p_2<-\eta q e^{\varrho t}<0.
\end{equation}
Recall that
$$
\begin{array}{lll}
\widehat{\mathcal H}(t,q, p_2)&=&\displaystyle \frac{1}{4\nu}e^{\varrho t}[\eta q +e^{-\varrho t}p_2]^2\\
&=&\displaystyle \frac{1}{4\nu}\left[e^{\varrho t} \eta^2 q^2 +2 \eta q p_2 +e^{-\varrho t} p_2^2\right]  \\
&=&  \displaystyle \frac{1}{4\nu}\left[\eta q[e^{\varrho t} \eta q + p_2]+ p_2[\eta q +p_2e^{-\varrho t^\prime} ]+p_2^2[e^{-\varrho t}-e^{-\varrho t^\prime}]\right].
\end{array}
$$
From the second inequality in Eq. (\ref{fg}),
we have   $e^{\varrho t} \eta q + p_2<0$.
From the first inequality in Eq. (\ref{fg}), we have
$\eta q+p_2 e^{-\varrho t^\prime}\ge \eta (q-q^\prime)$.
Since $p_2<0$, we  obtain
$$
p_2(\eta q+p_2 e^{-\varrho t^\prime})\le \eta p_2(q-q^\prime).
$$
Therefore,
\begin{equation}\label{kk0}
\begin{array}{lll}
\widehat{\mathcal H} (t, q, p_2) &\le&\displaystyle  \frac{1}{4\nu}\left[ \eta p_2(q-q^\prime)+p_2^2[e^{-\varrho t}-e^{-\varrho t^\prime}]\right]\\
&\le &\displaystyle  \frac{1}{4\nu}\left[ \eta |p_2|  |q-q^\prime|+\varrho |p_2|^2 |t-t^\prime|\right].
\end{array}
\end{equation}
The last inequality in Eq. (\ref{kk0}) is due to the result in Eq. (\ref{pol}).
Hence,  $\widehat{\mathcal H}$ satisfies the inequality
$$
\big|\widehat{\mathcal H}(t,q, p_2)-\widehat{\mathcal H}(t^\prime,q^\prime, p_2)\big|\le\frac{\eta}{4\nu} |p_2||q-q^\prime|+\frac{\rho}{4\nu}|p_2|^2|t-t^\prime|,
$$
which implies Inequality (\ref{ok23}).
\end{description}

\end{document}